%% file: article.tex
\documentclass[a4paper,twosided]{article}
\pdfoutput=1
\renewcommand{\pagebreak}{}
\usepackage{ifthen}
\newboolean{proceedings_version}
\setboolean{proceedings_version}{false}
\usepackage{amsfonts,amstext,amsmath,amssymb,mathrsfs,stmaryrd,calc,amsthm,xspace,mathtools,bold-extra,mathdots,enumerate}
\usepackage{bm}%
\providecommand{\Upgamma}{\Gamma}
\usepackage{macros}
\providecommand{\qedhere}{\qed}
\newtheorem*{ltheorem}{Length Function Theorem}
\newtheorem*{starnotation}{Notation}
\usepackage{thmtools}
\declaretheorem[numberwithin=section]{theorem}
\declaretheorem[sibling=theorem]{lemma}
\declaretheorem[sibling=theorem]{proposition}
\declaretheorem[sibling=theorem]{corollary}

\declaretheorem[sibling=theorem,style=definition]{definition}
\declaretheorem[sibling=theorem,style=definition]{example}
\declaretheorem[parent=theorem,style=remark]{claim}
\declaretheorem[sibling=theorem,style=remark]{remark}

\makeatletter
\def\ps@firstpage{\let\@mkboth\@gobbletwo%
\def\@oddhead{}%
\def\@oddfoot{%
\hbox{%
  \vbox to 100 pt{\footnotesize%
    \begin{minipage}[t]{\textwidth}\par An extended
      abstract of this work is published as Schmitz, S. and Schnoebelen, {\relax Ph}., 2011.
\newblock Multiply-recursive bounds with {H}igman's {L}emma.
\newblock In Aceto, L., Henzinger, M., and Sgall, J., editors, \emph{ICALP
  2011}\natconfdetails[\emph{38th}]{\emph{International Colloquium on Automata,
  Languages and Programming}}, volume 6756 of \emph{Lecture Notes in Computer
  Science}, pages 441--452. Springer.
\newblock \href{http://arxiv.org/abs/1103.4399}
  {\nolinkurl{doi:10.1007/978-3-642-22012-8_35}}.\end{minipage}}\hfill
    }}
\def\@evenhead{}
\def\@evenfoot{}}%
\makeatother
\usepackage{natbib}
\renewcommand{\cite}{\citep}\newcommand{\natconfdetails}[2][]{, #1 #2}
\newcommand{\ltheoremautorefname}{Length Function Theorem}
\renewcommand{\theoremautorefname}{Theorem}
\renewcommand{\propositionautorefname}{Proposition}
\renewcommand{\corollaryautorefname}{Corollary}
\renewcommand{\lemmaautorefname}{Lemma}
\renewcommand{\definitionautorefname}{Definition}

\newcommand{\sectionautorefname}{Section}

\def\UrlBreaks{\do\@\do\\\do\/\do\!\do\|\do\;\do\>\do\]%
 \do\,\do\?\do\'\do+\do\=\do\#}
\providecommand{\urlstyle}[1]{}
\providecommand{\doi}[1]{\href{http://dx.doi.org/#1}{\nolinkurl{doi:#1}}}
\usepackage{hyperref}
\def\UrlBreaks{\do\@\do\\\do\/\do\!\do\|\do\;\do\>\do\]%
 \do\,\do\?\do\'\do+\do\=\do\#}
\providecommand{\urlstyle}[1]{}
\providecommand{\doi}[1]{\href{http://dx.doi.org/#1}{\nolinkurl{doi:#1}}}
\begin{document}

\title{Multiply-Recursive Upper Bounds with \mbox{Higman's Lemma}}
\author{Sylvain Schmitz \and Philippe Schnoebelen}
\date{{\small
  LSV, ENS Cachan \& CNRS, Cachan, France\\
\href{mailto:sylvain.schmitz@lsv.ens-cachan.fr,phs@lsv.ens-cachan.fr}{\{schmitz,phs\}@lsv.ens-cachan.fr}}}
\maketitle
\begin{abstract}
We develop a new analysis for the length of controlled bad sequences in
well-quasi-orderings based on Higman's Lemma. This leads to tight
multiply-recursive upper bounds that readily apply to several verification
algorithms for well-structured systems.
\end{abstract}
\thispagestyle{firstpage}

\input{sec-intro}

\input{sec-nwqo}

\input{sec-algebra}
\input{sec-otypes}
\input{sec-hierarchies}
\input{sec-apps}

\input{sec-concl}

\bibliographystyle{../natbibsrt}
\bibliography{../journals,../mcaloon,higman}

\appendix\renewcommand{\refname}{Additional References}

\input{appendix}
\input{hierarchies}

\input{sec-comments}
\input{appendices_article.bbl}

\end{document}

%% file: sec-intro.tex
\section{Introduction}\label{sec-intro}

\paragraph{Well-quasi-orderings} (wqo's) are an important tool in logic and computer
science~\cite{kruskal72}. They are the key ingredient to a large number of
decidability (or finiteness, regularity, \ldots) results. In
constraint solving, automated deduction, program analysis, and many more
fields, wqo's usually appear under the guise of specific tools, like
Dickson's Lemma (for tuples of integers), Higman's Lemma (for words and their
subwords), Kruskal's Tree Theorem and its variants (for finite trees
with embeddings), and recently the Robertson-Seymour Theorem (for
graphs and their minors). In program verification, wqo's are the basis
for \emph{well-structured systems}~\cite{wqo,FinSch-wsts,henzinger2005}, a
generic framework for infinite-state systems.

\paragraph{Complexity.}
Wqo's are seldom used in complexity analysis. In order to extract
complexity upper bounds for an algorithm whose termination proof
rests on Dickson's or Higman's Lemma, one must be able to bound the length
of so-called ``controlled bad sequences'' (see \autoref{def-control}). Here
the available results are not very well known in computer science, and
their current packaging does not make them easy to read and apply. For
applications like the complexity of lossy channel
systems~\cite{CS-lics08} that rely on Higman's Lemma over $\Upgamma_p^\ast$ (the words over a
$p$-letter alphabet), what we really need is something like:
\begin{ltheorem}
\label{theo-in-introduction}
Let $\badd{\Upgamma_p^\ast}(n)$ be the maximal length of bad sequences
$w_0,w_1,w_2,\ldots$ over $\Upgamma_p^\ast$ with $p\geq 2$ s.t.\
$\len{w_i}<g^i(n)$ for $i=0,1,2,\ldots$ If the \emph{control function}
$g$ is primitive-recursive, then the \emph{length function}
$\badd{\Upgamma_p^\ast}$ is bounded by a function in
$\FGH{\omega^{p-1}}$.\footnote{Here the functions $F_\alpha$ are the
  ordinal-indexed levels of the Fast-Growing Hierarchy~\cite{fast},
  with multiply-recursive complexity starting at level
  $\alpha=\omega$, i.e., Ackermannian complexity, and stopping just
  before level $\alpha=\omega^\omega$, i.e., hyper-Ackermannian
  complexity.  The function classes $\FGH{\alpha}$ denote their
  elementary-recursive closure.}
\end{ltheorem}
\noindent
Unfortunately, the literature contains no such clear statement (see
the comparison with existing work below). 

\paragraph{Our Contribution.}
We provide a new and \emph{self-contained} proof of
the \ltheoremautorefname{}, a fundamental result that (we think) deserves a wide audience.
The exact statement we prove, \autoref{theo-main} below, is rather general: it is
parameterized by the control function $g$ and accomodates various
combinations of $\Upgamma_p^\ast$ sets without losing precision.
For this we significantly extend the setting we developed for Dickson's
Lemma~\cite{FFSS-arxiv10}: We rely on iterated residuations with a simple
but explicit algebraic framework for handling wqo's and their residuals in
a compositional way. Our computations can be kept relatively simple by
means of a fully explicit notion of ``normed reflection'' that captures the
over-approximations we use, all the while enjoying good algebraic
properties.
We also show \emph{how to apply} the \ltheoremautorefname{} by deriving precise
multiply-recursive upper bounds, parameterized by the alphabet size, for
the complexity of lossy channel systems and the Regular Post Embedding
Problem (see \autoref{sec-apps}).

\paragraph{Comparison with Existing Work.}
(Here, and for easier comparison, we assume that the control function $g$ is the successor
function.)

For $\Nat^k$ (i.e., Dickson's Lemma), \citeauthor{clote} gives an
explicit upper bound at level $\FGH{k+6}$ extracted from complex
Ramsey-theoretical results, hence hardly
self-contained~\cite{clote}. This is a simplification over an earlier
analysis by \citeauthor{mcaloon}, which leads to a uniform upper bound
at level $\FGH{k+1}$, but gives no explicit statement nor asymptotic
analysis~\cite{mcaloon}.  Both analyses are based on large intervals
and extractions, and \citeauthor{mcaloon}'s is technically quite
involved.  With D.\ and S.\ Figueira, we improved this to an explicit
and tight $\FGH{k}$~\cite{FFSS-arxiv10}.

For $\Upgamma_p^\ast$ (Higman's Lemma), \citeauthor{cichon98} exhibit a
reduction method, deducing bounds (for tuples of) words on $\Upgamma_p$ from
bounds on the $\Upgamma_{p-1}$ case~\cite{cichon98}. Their decomposition is
clear and self-contained, with the control function made explicit. It ends
up with some inequalities, collected
in~\cite[\sectionautorefname~8]{cichon98}, from which it is not clear
what precisely are the upper bounds one can extract.
Following this, \citeauthor{touzet2002} claims a bound of
$F_{\omega^p}$~\cite[\theoremautorefname~1.2]{touzet2002} with an
analysis based on iterated residuations but the proof (given
in~\cite{THESE-Touzet97}) is incomplete. %

Finally, \citeauthor{weiermann94} gives an $\FGH{\omega^{p-1}}$-like
bound for
$\Upgamma_p^\ast$~\cite[\corollaryautorefname~6.3]{weiermann94} for
sequences produced by term rewriting systems, but his analysis
is considerably more involved (as can be expected since it applies to
the more general Kruskal Theorem) and one cannot easily extract an
explicit proof for his \corollaryautorefname~6.3.

Regarding lower bounds, it is known that $F_{\omega^{p-1}}$ is
essentially tight~\citep{lowhigman}.

\paragraph{Outline of the Paper.}
All basic notions are recalled in \autoref{sec-nwqo}, leading to the
Descent Equation~\eqref{eq-descent}. Reflections in an algebraic
setting are defined in \autoref{sec-algebra}, then transfered in an
ordinal-arithmetic setting in \autoref{sec-otypes}.  We prove the Main
Theorem in \autoref{sec-hierarchies}, before illustrating its uses in
\autoref{sec-apps}.  
\ifthenelse{\boolean{proceedings_version}}{%
All the proofs missing from this extended abstract can be found at
\url{http://arxiv.org/abs/1103.4399}.
}{%
An appendix contains all the details omitted from the main text.
}

%% file: sec-nwqo.tex
\section{Normed Wqo's and Controlled Bad Sequences}
\label{sec-nwqo}

We recall some basic notions of wqo-theory~\citep[see e.g.][]{kruskal72}. A
\emph{quasi-ordering} (a ``qo'') is a relation $(A;\leq)$ that is reflexive
and transitive. As usual, we write $x<y$ when $x\leq y$ and $y\not\leq x$,
and we denote structures $(A;P_1,\ldots,P_m)$ with just the support set $A$
when this does not lead to ambiguities. Classically, the substructure
\emph{induced} by a subset $X\subseteq A$ is $(X;P_1{}_{|X},\ldots,P_m{}_{|X})$
where, for a predicate $P$ over $A$, $P_{|X}$ is its trace over $X$.

A qo $A$ is a \emph{well-quasi-ordering} (a ``wqo'') if every
infinite sequence $x_0,x_1,x_2,\ldots$ contains an infinite increasing
subsequence $x_{i_0}\leq x_{i_1}\leq x_{i_2}\cdots$ Equivalently, a qo is a
wqo if it is well-founded (has no infinite strictly decreasing sequences)
and contains no infinite antichains (i.e., set of pairwise incomparable
elements). Every induced substructure of a wqo is a wqo.

\paragraph{Wqo's With Norms.}
A \emph{norm function} over a set $A$ is a mapping
$\len{.}_A:A\rightarrow\Nat$ that provides every element of $A$ with a
positive integer, its \emph{norm}, capturing a notion of size. For
$n\in\Nat$, we let $A_{<n} \egdef \{x\in A~|~\len{x}_A<n\}$ denote the
subset of elements with norm below $n$. The norm function is said to be
\emph{proper} if $A_{<n}$ is finite for every $n$.
\begin{definition}
\label{def-nwqo}
A \emph{normed wqo} (a ``nwqo'') is a wqo $(A;\leq_A,\len{.}_A)$
equipped with a proper norm function.
\end{definition}

\noindent
There are no special conditions on norms, except being proper. In
particular no connection is required between the ordering of elements and
their norms. In applications, norms are related to natural complexity
measures.
\begin{example}[Some Basic Wqo's]
\label{ex-basic-wqos}
The set of natural numbers $\Nat$ with the usual ordering is the
smallest infinite wqo.
For every $p\in\Nat$, we single out two $p$-element
wqo's: $\seg{p}$ is the $p$-element initial segment of $\Nat$, i.e., the
set $\{0,1,2,\ldots,p-1\}$ ordered linearly, while $\Upgamma_p$ is the
$p$-letter alphabet $\{a_1,\ldots,a_p\}$ where distinct letters are
unordered.
We turn them into nwqo's by fixing the following:
\begin{align}
\label{eq-lens}
\len{k}_\Nat = \len{k}_{\seg{p}}   & \egdef k \:,
&
\len{a_i}_{\Upgamma_p}		   & \egdef 0 \:.
\end{align}
\end{example}

We write $A\equiv B$ when the two nwqo's $A$ and $B$ are \emph{isomorphic}
structures. For all practical purposes, isomorphic nwqo's can be
identified, following a standard practice that significantly simplifies the
notational apparatus we develop in \autoref{sec-algebra}.
For the moment, we only want to stress that, in particular, \emph{norm
  functions must be preserved} by nwqo isomorphisms.
\begin{example}[Isomorphism Between Basic Nwqo's]
On the positive side, $\seg{0}\equiv\Upgamma_0$ and also
$\seg{1}\equiv\Upgamma_1$ since $\len{a_1}_{\Upgamma_1}
= 0 = %
\len{0}_{\seg{1}}$.
By contrast $\seg{2}\not\equiv\Upgamma_2$: not only these two have
non-isomorphic order relations, they also have different norm functions.
\end{example}

\paragraph{Good, Bad, and Controlled Sequences.}
A sequence $\xxx=x_0,x_1,x_2,\ldots$ over a qo
is \emph{good} if $x_i\leq x_j$ for some positions $i<j$. It is
\emph{bad} otherwise. \emph{Over a wqo}, all infinite sequences
are good (equivalently, all bad sequences are finite).

We are interested in the maximal length of bad sequences for a given wqo.
Here, a difficulty is that, in general, bad sequences can be arbitrarily
long and there is no finite maximal length. However, in our applications we
are only interested in bad sequences generated by some algorithmic method,
i.e., bad sequences whose complexity is controlled in some way.

\begin{definition}[Control Functions and Controlled Sequences]
\label{def-control}
~\\
A \emph{control function} is a mapping $g:\Nat\rightarrow \Nat$.
For a \emph{size} $n\in\Nat$, a sequence $\xxx=x_0,x_1,x_2,\ldots$ over
a nwqo $A$ is \emph{$(g,n)$-controlled} $\equivdef$
\begin{equation*}
	       \forall i=0,1,2,\ldots:\; \len{x_i}_A < g^i(n)
	       \;=\; \obracew{g(g(\ldots g}{i \text{ times}}(n))) \: .
\end{equation*}
\end{definition}
\noindent
Why $n$ is called a ``size'' appears with
\autoref{prop-descent-equation} and its proof. A pair $(g,n)$ is just
called a \emph{control} for short. We say that a sequence $\xxx$ is
\emph{$n$-controlled} (or just \emph{controlled}), when $g$ (resp.\ $g$ and
$n$) is clear from the context. Observe that the empty sequence is always a
controlled sequence.

\ifthenelse{\boolean{proceedings_version}}{%
\begin{proposition}
}{%
\begin{proposition}[See App.~\ref{app-prop-L-finite}]
}
\label{prop-L-finite}
Let $A$ be a nwqo and $(g,n)$ a control.
There exists a finite maximal length $\bad\in\Nat$ for $(g,n)$-controlled
bad sequences over $A$.
\end{proposition}
We write $\bad_{A,g}(n)$ for this maximal length, a number that depends
on all three parameters: $A$, $g$ and $n$. However, for complexity
analysis, the relevant information is how, for given $A$ and $g$, the
\emph{length function} $\bad_{A,g}:\Nat\rightarrow\Nat$ behaves asymptotically,
hence our choice of notation. Furthermore, $g$ is a parameter that
remains fixed in our analysis and applications, hence it is usually left
implicit. \textbf{From now on we assume a fixed control function $g$} and
just write $\bad_A(n)$ for $\bad_{A,g}(n)$. We further assume that $g$ is
\emph{smooth} ($\equivdef$ $g(x+1)\geq g(x)+1\geq x+2$ for all $x$),
which is harmless for applications but simplifies computations like
\eqref{eq-eval-bad-basic}.

\paragraph{Residuals Wqo's and a Descent Equation.}
Via residuals one expresses the length function by induction over nwqo's.
\begin{definition}[Residuals]
For a nwqo $A$ and an element $x\in A$, the \emph{residual} $A/x$ is the
substructure (a nwqo) induced by the subset $A/x\egdef\{y\in A~|~x\not\leq
y\}$ of elements that are not above $x$.
\end{definition}
\begin{example}[Residuals of Basic Nwqo's]
\label{ex-residuals}
For all $k<p$ and $i=1,\ldots,p$:
\begin{xalignat}{2}
\label{eq-ex-residuals}
\Nat/k &= \seg{p}/k = \seg{k}
\: ,
&
\Upgamma_{p}/a_i&\equiv \Upgamma_{p-1}
\:.
\end{xalignat}
\end{example}

\ifthenelse{\boolean{proceedings_version}}{%
\begin{proposition}[Descent Equation]
}{%
\begin{proposition}[Descent Equation, See App.~\ref{app-prop-descent-equation}]
}
\label{prop-descent-equation}
\begin{equation}
\label{eq-descent}
       \bad_{A}(n) = \max_{x\in A_{<n}} \bigl\{1 + \bad_{A/x}(g(n))\bigr\}		 \:.
\end{equation}
\end{proposition}
This reduces the $\bad_A$ function to a finite combination of $\bad_{A_i}$'s
where the $A_i$'s are residuals of $A$, hence ``smaller'' sets.
Residuation is well-founded for wqo's: a sequence of successive
residuals $A \supsetneq A/x_0 \supsetneq A/x_0/x_1 \supsetneq A/x_0/x_1/x_2 \supsetneq \cdots$ is
necessarily finite since $x_0,x_1,x_2,\ldots$ must be a bad sequence. Hence
the recursion in the Descent Equation is well-founded and can be used to
evaluate $\bad_A(n)$. This is our starting point for analyzing the
behaviour of length functions.

For example, using induction and Eq.~\eqref{eq-ex-residuals},  the Descent
Equation leads to:%
\begin{align}
\label{eq-eval-bad-basic}
\bad_{\Upgamma_p}(n) &= p \: ,
&
\bad_\Nat(n) &= n \: ,
&
\bad_\seg{p}(n) &= \min (n,p)
\; .
\end{align}

%% file: sec-algebra.tex
\section{An Algebra of Normed Wqo's}
\label{sec-algebra}

The algebraic framework we now develop has two main
goals. First
 it provides a \emph{notation} for denoting the wqo's  encountered in
algorithmic applications. These wqo's and their norm functions abstract
data structures that are built inductively by combining some basic wqo's.
Second, it supports a \emph{calculus} for the kind of
compositional computations, based on the Descent Equation, we
develop next.%

The
constructions we use in this paper are disjoint sums, cartesian products,
and Kleene stars (with Higman's order). These constructions are classic.
Here we also have to define how they combine the norm functions:
\begin{definition}[Sums, Products, Stars Nwqo's]
\label{def-nwqo-combinators}
The \emph{disjoint sum} $A_1+A_2$ of two nwqos $A_1$ and $A_2$ is
the nwqo given by%
\begin{align}
\label{eq-supp-sum}
A_1+A_2 &= \{\tup{i,x} ~|~ i\in\{1,2\} \text{ and } x\in A_i \}		       \:,
\\
\label{eq-order-sum}
\tup{i,x}\leq_{A_1+A_2}\tup{j,y} &\equivdef i = j \text{ and } x\leq_{A_i}y	\:,
\\
\label{eq-len-sum}
\len{\tup{i,x}}_{A_1+A_2} &\egdef \len{x}_{A_i}		\:.
\shortintertext{%
The \emph{cartesian product} $A_1\times A_2$ of two nwqos $A_1$
and $A_2$ is the nwqo given by
}%
\label{eq-supp-prod}
A_1\times A_2 &= \{\tup{x_1,x_2} ~|~ x_1\in A_1, x_2\in A_2\}		     \:,
\\
\label{eq-order-prod}
\tup{x_1,x_2}\leq_{A_1\times A_2}\tup{y_1,y_2} &\equivdef x_1\leq_{A_1} y_1 \text{ and } x_2\leq_{A_2}y_2     \:,
\\
\label{eq-len-prod}
\len{\tup{x_1,x_2}}_{A_1\times A_2} &\egdef \max(\len{x_1}_{A_1},\len{x_2}_{A_2})	  \:.
\shortintertext{%
The \emph{Kleene star} $A^\ast$ of a nwqo $A$
is the nwqo given by
}%
\label{eq-supp-star}
A^\ast \egdef \text{ all}&\text{ finite lists $(x_1\ldots x_n)$ of elements of $A$}		\:,
\\
\label{eq-order-star}
(x_1\ldots x_n)\leq_{A^\ast} (y_1\ldots y_m)
&\equivdef
\left\{
\begin{array}{l}
x_1\leq_A y_{i_1}
\wedge \cdots  \wedge
x_n\leq_A y_{i_n}
\\
\text{for some } 1\leq i_1<i_2<\cdots<i_n\leq m
\end{array}
\right.%
\:,
\\
\label{eq-len-star}
\len{(x_1\ldots x_n)}_{A^\ast} &\egdef \max(n,\len{x_1}_{A},\ldots,\len{x_n}_{A})	       \:.
\end{align}
\end{definition}

\noindent
It is well-known (and plain) that $A_1+A_2$ and $A_1\times A_2$ are indeed
wqo's when $A_1$ and $A_2$ are. The fact that $A^\ast$ is a wqo when $A$ is,
is a classical result called Higman's Lemma.
We let the reader check that the norm functions defined in Equations
\eqref{eq-len-sum}, \eqref{eq-len-prod}, and \eqref{eq-len-star}, are
proper and turn $A_1+A_2$, $A_1\times A_2$ and $A^\ast$ into nwqo's.
Finally, we note that nwqo isomorphism is a congruence for sum, product and
Kleene star.
\begin{starnotation}[$\bm{0}$ and $\bm{1}$]
We let $\bm{0}$ and $\bm{1}$ be short-hand notations for, respectively,
$\Upgamma_0$ (the empty nwqo) and $\Upgamma_1$ (the singleton nwqo with the $0$
norm).
\end{starnotation}

\noindent
This is convenient for writing down the following algebraic properties:
\ifthenelse{\boolean{proceedings_version}}{%
\begin{proposition}
}{%
\begin{proposition}[See App.~\ref{ax-iso-wqos}]
}
\label{prop-iso-wqos}
The following isomorphisms hold:\vspace*{-.5em}
\begin{xalignat*}{2}
A+ B &\equiv B+A
\: ,
&
A+(B+C) &\equiv (A+B)+C
\: ,
\\
A\times B &\equiv B\times A
\: ,
&
A\times (B\times C) &\equiv (A\times B)\times C
\: ,
\\
\bm{0}+A &\equiv A
\: ,
&
\bm{1}\times A &\equiv A
\: ,
\\
\bm{0}\times A &\equiv\bm{0}
\: ,
&
(A+A') \times B &\equiv (A\times B)+(A'\times B)
\: ,
\\
\bm{0}^\ast &\equiv \bm{1}
\: ,
&
\bm{1}^\ast &\equiv \Nat
\: .
\end{xalignat*}
\end{proposition}

\noindent
In view of these properties, we freely write $A\cdot k$ and $A^k$ for the
$k$-fold sums and products $A+\cdots+A$ and $A\times \cdots \times A$.
Observe that $A\cdot k\equiv A\times\Upgamma_k$.

\paragraph{Reflecting Normed Wqo's.}
Reflections are the main comparison/abstraction tool we shall use.
They let us simplify instances of the Descent Equation by replacing
all $A/x$ for $x\in A_{<n}$ by a single (or a few)  $A'$ that
is smaller than $A$ but large enough to reflect all considered $A/x$'s.

\begin{definition}
\label{def-reflection}
A \emph{nwqo reflection} is a mapping $h:A\rightarrow
B$ between two nwqo's that satisfies the two following properties:
\begin{align*}
\forall x,y\in A: \;
& h(x)\leq_B h(y) \text{ implies } x\leq_A y		\:,
\\
\forall x\in A: \;
& \len{h(x)}_B\leq \len{x}_A		\:.
\end{align*}
\end{definition}
In other words, a nwqo reflection is an order reflection that is also norm-decreasing (not
necessarily strictly).

We write $h:A\hookrightarrow B$ when $h$ is a nwqo reflection and say that
\emph{$B$ reflects $A$}. This induces a relation between nwqos, written
$A\hookrightarrow B$.

Reflection is transitive since $h:A\hookrightarrow B$ and
$h':B\hookrightarrow C$ entails $h'\circ h:A\hookrightarrow C$. It is also
reflexive, hence reflection is a quasi-ordering. Any nwqo reflects its
substructures since $\textit{Id}:X\hookrightarrow A$ when $X$ is a
substructure of $A$. Thus $\bm{0}\hookrightarrow A$ for any $A$, and
$\bm{1}\hookrightarrow A$ for any non-empty $A$.

\begin{example}
\label{ex-basic-residuals}
Among the basic nwqos from \autoref{ex-basic-wqos}, we note the
following relations (or absences thereof). For any $p\in\Nat$,
$\seg{p}\hookrightarrow\Upgamma_p$, while
$\Upgamma_p\not\hookrightarrow\seg{p}$ when $p\geq 2$.
The reflection of substructures yields $\seg{p}\hookrightarrow\Nat$ and
$\Upgamma_p\hookrightarrow\Upgamma_{p+1}$. Obviously,
$\Nat\not\hookrightarrow\seg{p}$ and
$\Upgamma_{p+1}\not\hookrightarrow\Upgamma_p$.
\end{example}

\noindent
Reflections preserve controlled bad sequences.
Let $h:A\hookrightarrow B$, consider a sequence
$\xxx=x_0,x_1,\ldots,x_l$ over $A$, and write $h(\xxx)$ for
$h(x_0),h(x_1),\ldots,h(x_l)$, a sequence over $B$. Then $h(\xxx)$ is bad
when $\xxx$ is, and $n$-controlled when $\xxx$ is. Hence:
\begin{equation}
\label{eq-majo-prop}
A\hookrightarrow B
\text{ implies }
\bad_{A}(n)\leq\bad_{B}(n) \text{ for all $n$}
\:.
\end{equation}
\noindent
Reflections are compatible with product, sum, and Kleene star.
\ifthenelse{\boolean{proceedings_version}}{%
\begin{proposition}[Reflection is a Preconguence]
}{%
\begin{proposition}[Reflection is a Preconguence, see App.~\ref{app-prop-precongr}]
}
\label{prop-precongr}
\begin{align}
A\hookrightarrow A' \;\text{and}\;
B\hookrightarrow B'
&\;\text{imply}\;
A+ B\hookrightarrow A'+ B' \;\text{and}\;
A\times B\hookrightarrow A'\times B'
\: ,
\\
A\hookrightarrow A'
&\;\text{implies}\;
A^\ast \hookrightarrow A'^\ast
\: .
\end{align}
\end{proposition}

\paragraph{Computing and Reflecting Residuals.}
We may now tackle our first main problem: computing residuals $A/x$. This
is done by induction over the structure of $A$.
\ifthenelse{\boolean{proceedings_version}}{%
\begin{proposition}[Inductive Rules For Residuals]
}{%
\begin{proposition}[Inductive Rules For Residuals, see App.~\ref{app-prop-refl-residuals}]
}
\label{prop-refl-residuals}
\begin{align}
\label{eq-residual-sum}
(A+B)/\tup{1,x} = (A/x)+B
\: ,
&\quad
(A+B)/\tup{2,x} = A+(B/x)
\: ,
\\
\label{eq-residual-prod}
(A\times B)/\tup{x,y} &\hookrightarrow
\bigl[(A/x)\times B\bigr]+\bigl[A\times(B/y)\bigr]
\: ,
\\
\label{eq-residual-star}
A^\ast/(x_1\ldots x_n) &\hookrightarrow \Upgamma_n\times A^n\times (A/x_1)^\ast\times
\cdots \times (A/x_n)^\ast\:,\\
\label{eq-residual-finite-star}
\Upgamma_{p+1}^\ast/(x_1\ldots x_n)&\hookrightarrow \Upgamma_n\times (\Upgamma_p^\ast)^{n}\;.
\end{align}
\end{proposition}

\noindent
Equation~\eqref{eq-residual-finite-star} is a refinement of
\eqref{eq-residual-star} in the case of finite alphabets%
.
Since it provides reflections instead of isomorphisms,
\autoref{prop-refl-residuals} is not meant to support exact
computations of $A/x$ by induction over the structure of $A$. More to
the point, it yields over-approximations that are sufficiently precise
for our purposes while bringing important simplifications when we have
to reflect (the max of) all $A/x$ for all $x\in A_{<n}$.

%
%
%

%
%
%
%
%

%

%% file: sec-otypes.tex
%

\section{Reflecting Residuals in Ordinal Arithmetic}
\label{sec-otypes}

We now introduce an \emph{ordinal} notation for nwqo's. The purpose is
twofold.  Firstly, the ad-hoc techniques we use for evaluating,
reflecting, and comparing residual nwqo's are more naturally stated
within the language of ordinal arithmetic.  Secondly, these ordinals
will be essential for bounding $\bad_A$ using functions in
subrecursive hierarchies.  For these developments, we restrict
ourselves to \emph{exponential} nwqo's, i.e., nwqo's obtained from
finite $\Upgamma_p$'s with sums, products, and \emph{Kleene star
  restricted to the $\Upgamma_p$'s}.  Modulo isomorphism,
$\Nat^k\equiv\prod_{i=1}^k\Upgamma_1^\ast$ is exponential.

\paragraph{Ordinal Terms.}
We use Greek letters like $\alpha,\beta,\ldots$ to denote ordinal
terms in Cantor Normal Form (CNF) built using $0$, addition, and
$\omega$-exponentiation (we restrict ourselves to ordinals
$<\ezero$).  A term $\alpha$ has the general form
$\alpha=\omega^{\beta_1}+\omega^{\beta_2}+\cdots+\omega^{\beta^m}$
with $\beta_1\geq\beta_2\geq \cdots \geq\beta_m$ (ordering defined
below) and where we distinguish between three cases: $\alpha$ is $0$
if $m=0$, $\alpha$ is a \emph{successor} if ($m>0$ and) $\beta_m=0$,
$\alpha$ is a \emph{limit} if $\beta_m\not=0$ (in the following,
$\lambda$ will always denote a limit, and we write $\alpha+1$ rather
than $\alpha+\omega^0$ for a successor).  We say that $\alpha$ is
principal (additive) if $m=1$.

Ordering among our ordinals is defined inductively by
\begin{equation}
\label{eq-def-ordering}
\alpha < \alpha'
\equivdef
\left\{\begin{array}{l}
\text{$\alpha=0$ and $\alpha'\not=0$, or}
\\
\text{$\alpha=\omega^\beta+\gamma$, $\alpha'=\omega^{\beta'}+\gamma'$
  and}
\left\{\begin{array}{l}
\text{$\beta<\beta'$, or}
\\
\text{$\beta=\beta'$ and $\gamma<\gamma'$.}
\end{array}\right.%
\end{array}\right.%
\end{equation}
We let $\CNF[\alpha]$ denote the set of ordinal terms ${<}\alpha$.

For $c\in\Nat$, $\omega^\beta\cdot c$ denotes the $c$-fold addition
$\omega^\beta + \cdots + \omega^\beta$.  We sometimes write terms
under a ``strict'' form $\alpha= \omega^{\beta_1}\cdot
c_1+\omega^{\beta_2}\cdot c_2+\cdots+\omega^{\beta^m}\cdot c_m$ with
$\beta_1>\beta_2>\cdots>\beta_m$, where the $c_i$'s, called
\emph{coefficients}, must be $>0$.

Recall the definitions of the \emph{natural sum} $\alpha\oplus\alpha'$
and \emph{natural product} $\alpha\otimes\alpha'$ of two terms in
$\CNF[\ezero]$:
\begin{align*}
  \sum_{i=1}^m\omega^{\beta_i}\oplus\sum_{j=1}^n\omega^{\beta'_j}&\eqdef\sum_{k=1}^{m+n}\omega^{\gamma_k}\:,&
  \sum_{i=1}^m\omega^{\beta_i}\otimes\sum_{j=1}^n\omega^{\beta'_j}&\eqdef\bigoplus_{i=1}^m\bigoplus_{j=1}^n\omega^{\beta_i\oplus\beta'_j}
\:,
\end{align*}
where $\gamma_1\geq\cdots\geq\gamma_{m+n}$ is a rearrangement of
$\beta_1,\ldots,\beta_m,\beta'_1,\ldots,\beta'_n$.
For $\alpha\in\CNF$, the decomposition
$\alpha=\sum_{i=1}^m\omega^{\beta_i}$ uses $\beta_i$'s that are in
$\CNF[\omega^\omega]$, i.e., of the form
$\beta_i=\sum_{j=1}^{k_i}\omega^{p_{i,j}}$ (with each
$p_{i,j}<\omega$) so that $\omega^{\beta_i}$ is
$\bigotimes_{j=1}^{k_i}\omega^{\omega^{p_{i,j}}}$. A term
$\omega^{\omega^{p}}$ is called a \emph{principal multiplicative}.

We map exponential nwqo's to ordinals in $\CNF$ using their
\emph{maximal order type}~\citep{dejongh77}.  Formally $o(A)$ is
defined by
\begin{align}
\label{eq-odef-1}
  o(\Upgamma_p)&\eqdef p \:, 
& o(\Upgamma_0^\ast)&\eqdef\omega^0 \:,
& o(\Upgamma_{p+1}^\ast)&\eqdef\omega^{\omega^p} \:,
\end{align}\vspace*{-1em}
\begin{align}
\label{eq-odef-2}
 o(A+B)&\eqdef o(A)\oplus o(B) \:,
&o(A\times B)&\eqdef o(A)\otimes o(B) \:.
\end{align}
Conversely, there is a \emph{canonical exponential nwqo} $C(\alpha)$ for
each $\alpha$ in $\CNF$:
\begin{align}
\label{eq-Cdef}
C\Bigl(\omega^{\beta_1} +\cdots +\omega^{\beta_m}\Bigr)
=
	C\Bigl(\bigoplus_{i=1}^m\bigotimes_{j=1}^{k_i}\omega^{\omega^{p_{i,j}}}\Bigr)
\egdef
		 \sum_{i=1}^m\prod_{j=1}^{k_i}\Upgamma_{(p_{i,j}+1)}^\ast
\: .
\end{align}
Then, $o$ and $C$ are bijective inverses (modulo isomorphism of nwqo's),
compatible with sums and products\ifthenelse{\boolean{proceedings_version}}{%
.
}{%
 (see App.~\ref{app-comments}).
}
This
correspondence equates between terms that, on one side, denote partial
orderings with norms, and on the other side, ordinals in $\CNF$.

\paragraph{Derivatives.}
We aim to replace the ``all $A/x$ for $x\in A_{<n}$'' by a computation
of ``some derived $\alpha'\in\partial_n\alpha$'' where $\alpha=o(A)$,
see \autoref{theo-deriv-refl} below. For this purpose, the definition
of derivatives is based on the inductive rules in
\autoref{prop-refl-residuals}.

Let $n>0$ be some norm. We start with principal ordinals and define
\begin{align}
  \label{eq-dn}
  D_n\!\!\left(\omega^{\omega^p}\right)
  &\eqdef\begin{cases}
    n-1&\text{if $p=0$,}\\
    \omega^{(\omega^{p-1}\cdot (n-1))}\cdot (n-1)&\text{otherwise.}
  \end{cases}
\\
  \label{eq-Dn}
  D_n\left(\omega^{\omega^{p_1}+\cdots+\omega^{p_k}}\right)
  &\eqdef\bigoplus_{j=1}^k\left(D_n\!\!\left(\omega^{\omega^{p_j}}\right)\otimes\bigotimes_{\ell\neq
    j}\omega^{\omega^{p_\ell}}\right) \:.
\end{align}
Now, with any $\alpha\in\CNF$, we associate the set of its \emph{derivatives}
$\partial_n\alpha$ with
\begin{align}
\label{eq-partialn}
\partial_n\Bigl(\sum^m_{i=1}\omega^{\beta_i}\Bigr)
&\egdef
\Bigl\{
D_n\bigl(\omega^{\beta_i}\bigr)\oplus\sum_{\ell\not=i}\omega^{\beta_\ell}
~\Big\lvert~ i = 1, \ldots, m
\Bigr\} \: .
\end{align}

\noindent
This yields, for example, and assuming $p,k>0$:
\begin{align}
\label{eq-deriv-1}
\!\!\!\!D_n(1) & = 0
,
&
\!\!\!\!D_n(\omega) & = n-1
,
&
\!\!\!\!D_n\bigl(\omega^{\omega^{p}\cdot k}\bigr) &= \omega^{[\omega^{p}\cdot (k-1)+\omega^{p-1}\cdot
    (n-1)]}\cdot k(n-1)
\,,
\end{align}\vspace{-1em}
\begin{align}
\label{eq-deriv-2}
\!\!\!\partial_n 0 & =\emptyset,
&
\!\!\!\!\partial_n 1 & = \{ 0 \},
&
\!\!\!\!\partial_n \omega & = \{ n-1 \},
&
\!\!\!\!\partial_n (\omega^\beta\cdot (k+1))
& = \{\omega^\beta\cdot k \oplus D_n(\omega^\beta)\}.\!\!
\end{align}
Thus $\partial_n \alpha$ can be a singleton even when $\alpha$ is not
principal, e.g., $\partial_n (p+1)=\{p\}$. 
We sometimes write $\alpha\mathrel{\partial_n}\alpha'$ instead
of $\alpha'\in\partial_n \alpha$, seeing
 $\partial_n$  as
a relation. 
Note that $\partial_n\alpha \subseteq\CNF[\alpha]$\ifthenelse{\boolean{proceedings_version}}{%
,
}{%
 (see
App.~\ref{app-deriv-is-wf}), 
}
hence
$\partial\egdef\bigcup_{n<\omega}\mathrel{\partial_n}$ is
well-founded.

\ifthenelse{\boolean{proceedings_version}}{%
\begin{theorem}[Reflection by Derivatives]
}{%
\begin{theorem}[Reflection by Derivatives, see App.~\ref{app-theo-deriv-refl}]
}
\label{prop-deriv-major}	%
\label{theo-deriv-refl}
Let $x\in A_{<n}$ for some exponential $A$.
Then there exists $\alpha'\in\partial_n o(A)$ s.t.\ $A/x\hookrightarrow C(\alpha')$.
\end{theorem}
\noindent
Combining with equations
\eqref{eq-descent} and \eqref{eq-majo-prop}, we obtain:
\begin{equation}
\label{eq-main-inequality}
  \bad_{C(\alpha)}(n)\leq
  \max_{\alpha'\in\partial_{n}\alpha}\bigl\{1+\bad_{C(\alpha')}(g(n))\bigr\}\;.
\end{equation}

%
%
%
%
%

%

%% file: sec-hierarchies.tex
\section{Classifying $\bad$ using Subrecursive
  Hierarchies}\label{sec-hierarchies}
For $\alpha$ in
$\mathrm{CNF}(\omega^{\omega^\omega})$, define
\begin{equation}\label{eq:ubound}
  \uboundd{\alpha}(n)\eqdef\max_{\alpha'\in\partial_{n}\alpha}\left\{1+\uboundd{\alpha'}(g(n))\right\} \:.
\end{equation}
(Recall that $\partial$ is well-founded, thus
\eqref{eq:ubound} is well-defined). Comparing with
\eqref{eq-main-inequality}, we see that $\ubound_\alpha$ bounds the
length function: $\ubound_\alpha(n)\geq\badd{C(\alpha)}(n)$.

This defines an ordinal-indexed family of functions
$(\uboundd{\alpha})_{\alpha\in\mathrm{CNF}(\omega^{\omega^\omega})}$
similar to some classical subrecursive hierarchies, with the added
twist of the $\max$ operation---see \citep{buchholtz94,moser03} for
somewhat similar hierarchies. This is a real issue and one cannot
replace a ``$\max_{\alpha\in\ldots}\{\ubound_{\alpha}(x)\}$'' with
``$\ubound_{\sup\{\alpha\in\ldots\}}(x)$'' since $\ubound_{\alpha}$ is
not always bounded by $\ubound_{\alpha'}$ when $\alpha<\alpha'$.
E.g., $\ubound_{n+2}(n)=n+2>\ubound_{\omega}(n)=n+1$.

\paragraph{Subrecursive Hierarchies}
have been introduced as generators of classes
of functions.  For instance, writing $\FGH{\alpha}$ for the class of
functions elementary-recursive in the function $F_{\alpha}$ of the
\emph{fast growing hierarchy}, we can characterize the set of
primitive-recursive functions as $\bigcup_{k<\omega}\FGH{k}$, or that of
multiply-recursive functions as
$\bigcup_{\beta<\omega^{\omega}}\FGH{\beta}$~\citep{fast}.

Let us introduce (slight generalizations of) several classical
hierarchies from \citep{fast,cichon98}.  Those hierarchies are
defined through assignments of \emph{fundamental sequences}
$(\lambda_x)_{x<\omega}$ for limit ordinals $\lambda<\ezero$, verifying $\lambda_x<\lambda$ for all $x$ and
$\lambda=\sup_x\lambda_x$.  A standard assignment is defined by:
\begin{align}
\label{eq-fund}
  (\gamma+\omega^{\beta+1})_x&\eqdef\gamma+\omega^\beta\cdot (x+1)\,,&
  (\gamma+\omega^\lambda)_x&\eqdef\gamma+\omega^{\lambda_x}\;,
\end{align}
where $\gamma$ can be $0$.  Note that, in particular,
   $\omega_x=x+1$.  Given an assignment of fundamental sequences, one
   can define the ($x$-indexed) \emph{predecessor} $P_x(\alpha)<\alpha$
   of an ordinal $\alpha\not=0$ as
\begin{align}
\label{eq-pred}
 P_x(\alpha+1)&\eqdef\alpha\,,&P_x(\lambda)&\eqdef P_x(\lambda_x)\;.
\end{align}
Given a fixed smooth control function $h$, the \emph{Hardy hierarchy} $(h^\alpha)_{\alpha<\ezero}$ is
then defined by
\begin{align}\label{eq-hardy-h}
  \!\!h^0(x)&\eqdef x,&\!\!h^{\alpha+1}(x)&\eqdef
  h^\alpha(h(x)),&\!\!h^\lambda(x)&\eqdef h^{\lambda_{x}}(x)\,.
  \shortintertext{A closely related hierarchy is the \emph{length
      hierarchy} $(h_{\alpha})_{\alpha<\ezero}$ defined by}\label{eq-length-h}
  \!\!h_0(x)&\eqdef 0,&\!\!h_{\alpha+1}(x)&\eqdef
  1+h_{\alpha}(h(x)),&\!\!h_{\lambda}(x)&\eqdef
  h_{\lambda_{x}}(x)\,.
  \shortintertext{Last of all, the \emph{fast growing hierarchy} 
$(f_{\alpha})_{\alpha<\ezero}$ is
    defined through}\label{eq-fast-h}
  \!\!f_0(x)&\eqdef h(x),&\!\!f_{\alpha+1}(x)&\eqdef
  f_{\alpha}^{\omega_x}(x),&\!\!f_{\lambda}&\eqdef
  f_{\lambda_{x}}(x)\,.
\end{align}
Standard versions of these hierarchies are usually defined by setting $h$
as the successor function, in which case they are denoted $H^\alpha$,
$H_\alpha$, and $F_\alpha$ resp.
\ifthenelse{\boolean{proceedings_version}}{%
\begin{lemma}
}{%
\begin{lemma}[\citep{cichon83,cichon98} or App.~\ref{ax-hierarchies}] 
}
\label{lem-hierarchies}
For all $\alpha\in\CNF$ and $x\in\Nat$,\hfill%
  \begin{enumerate}
  \item\label{lem-hierarchies-1}
	$h_{\alpha}(x)=1+h_{P_{x}(\alpha)}(h(x))$ when $\alpha>0$,
  \item\label{lem-hierarchies-2}
	$h_{\alpha}(x)\leq h^\alpha(x)-x$,
  \item\label{lem-hierarchies-3} 
	$h^{\omega^\alpha\cdot r}(x)=f^r_{\alpha}(x)$ for all $r<\omega$,
  \item\label{lem-hierarchies-4}
	if $h$ is eventually bounded by $F_{\gamma}$, then
        $f_{\alpha}$ is eventually bounded by $F_{\gamma+\alpha}$.
  \end{enumerate}
\end{lemma}

\paragraph{Bounding the Length Function.}
Item~\ref{lem-hierarchies-1} of \autoref{lem-hierarchies} shows that
$\uboundd{\alpha}$ and $h_\alpha$ have rather similar expressions,
based on derivatives for $\uboundd{\alpha}$ and predecessors for $h_\alpha$;
they are in fact closely related:
\ifthenelse{\boolean{proceedings_version}}{%
\begin{proposition}
}{%
\begin{proposition}[See App.~\ref{ax-hardy}]
}
\label{prop-hardy}
  For all $\alpha$ in $\CNF$, there is a constant $k$ s.t.\ for all $n>0$,
  $\uboundd{\alpha,g}(n)\leq h_{\alpha}(kn)$ where $h(x)\eqdef x\cdot g(x)$.
\end{proposition}
Proposition~\ref{prop-hardy} translates for $n,p>0$ into an
\begin{equation}\label{eq-bad-hardy}
  \badd{\Upgamma_{p}^\ast,g}(n)\leq
  h_{\omega^{\omega^{p-1}}}((p-1)n)
\;\;\;\text{ for } h(x)\eqdef x\cdot g(x)
\end{equation}
upper bound on bad $(g,n)$-controlled sequences in $\Upgamma_p^\ast$.
We believe \eqref{eq-bad-hardy} answers a
wish expressed by \citeauthor{cichon98} in their
conclusion~\citep{cichon98}: ``an appropriate bound should be given by
the function $h_{\omega^{\omega^{p-1}}}$, for some reasonable $h$.''
\medskip

It remains to translate the bound of \autoref{prop-hardy} into a more
intuitive and readily usable one.  Combined with items 2--4 of
\autoref{lem-hierarchies}, \autoref{prop-hardy} allows us to state a
fairly general result in terms of the $(\FGH{\alpha})_\alpha$ classes
in the two most relevant cases (of which both the
\ltheoremautorefname{} given in the introduction and, if $\gamma\geq
2$, the $\FGH{\gamma+k}$ bound given for $\Nat^k$
in~\citep{FFSS-arxiv10}, are consequences):
\begin{theorem}[Main Theorem]
\label{theo-main}
Let $g$ be a smooth control function eventually bounded by a function
in $\FGH{\gamma}$, and let $A$ be an exponential nwqo with maximal
order type $<\omega^{\beta+1}$.  Then $\badd{A,g}$ is bounded by a
function in
\begin{itemize}
\item $\FGH{\beta}$ if $\gamma<\omega$ (e.g.\ if $g$ is
  primitive-recursive) and $\beta\geq\omega$,
\item $\FGH{\gamma+\beta}$ if $\gamma\geq 2$ and $\beta<\omega$.
\end{itemize}
\end{theorem}

%% file: sec-apps.tex
\section{Refined Complexity Bounds for Verification Problems}
\label{sec-apps}

This section provides two \emph{examples} where our Main Theorem leads
to precise multiply-recursive complexity upper bounds for problems
that were known to be decidable but not primitive-recursive. Our
choice of examples is guided by our close familiarity with these
problems (in fact, they have been our initial motivation for looking
at subrecursive hierarchies) and by their current role as master
problems for showing Ackermann complexity lower bounds in several
areas of verification. (A more explicit vademecum for potential users
of the Main Theorem can be found in~\cite{FFSS-arxiv10}.)

\paragraph{Lossy Channel Systems.}
The wqo associated with a lossy channel system $S=(Q,M,C,\Delta)$ is
the set $A_S\egdef Q\times (M^\ast)^C$ of its configurations, ordered
with embedding (see details in~\cite{CS-lics08}). Here $Q$ is a set of
$q$ control locations, $M$ is a size-$m$ message alphabet and $C$ is a
set of $c$ channels. Hence, we obtain $A_S \equiv
q\cdot\bigl(\Upgamma_m^\ast\bigr)^c$. For such lossy
systems~\cite{phs-rp2010}, reachability, safety and termination can be
decided by algorithms that only need to explore bad sequences over
$A_S$. In particular, $S$ has a non-terminating run from configuration
$s_\text{init}$ iff it has a run of length
$\bad_{A_S}(\len{s_\text{init}})$, and the shortest run (if one
exists) reaching $s_\text{final}$ from $s_\text{init}$ has length at
most $\bad_{A_S}(\len{s_\text{final}})$. Here the sequences (runs of
$S$, forward or backward) are controlled with $g=\textit{Succ}$.  Now,
since $o(A_S)=\omega^{(\omega^{m-1}\cdot c)}\cdot q$,
\autoref{theo-main} gives an overall complexity at level
$\FGH{\omega^{(m-1)}\cdot c}$, which is the most precise upper bound
so far for lossy channel systems.

Regarding lower bounds, the construction in~\cite{CS-lics08} proves a
$\FGH{\omega^{K}}$ lower bound for systems using $m=K+2$ different
symbols, $c=2$ channels, and a quadratic $q\in O(K^2)$ number of
states. If emptiness tests are allowed (an harmless extension for
lossy systems, see~\cite{phs-rp2010}) one can even get rid of the
$\texttt{\#}$ separator symbol in that construction (using more
channels instead) and we end up with $m=K+1$ and $c=4$. Thus the
demonstrated upper and lower bounds are very close, and tight when
considering the recursivity-multiplicity level.

\paragraph{$\PEPr$, the Regular Post Embedding Problem,}
is an abstract problem that relaxes Post's Correspondence Problem by
replacing the equality ``$u_{i_1}\dots u_{i_n}=v_{i_1}\dots v_{i_n}$'' with
embedding ``$u_{i_1}\dots u_{i_n}\leq_{\Upgamma^\ast} v_{i_1}\dots
v_{i_n}$'' (all this under a ``$\exists i_1,\ldots,i_n$ in some regular
$R$'' quantification). It was introduced in~\cite{CS-fsttcs07} where
decidability was shown thanks to Higman's Lemma. Non-trivial reductions
between $\PEPr$ and lossy channel systems exist. Due to its abstract
nature, $\PEPr$ is a potentially interesting master problem for proving
hardness at multiply-recursive and hyper-Ackermannian, i.e.,
$\FGH{\omega^\omega}$, levels (see refs in~\cite{CS-countpep}).

A pumping lemma was proven in~\cite{CS-countpep}, which relies on the
$\bad_{A}$ function, and from which we can now derive more precise
complexity upper bounds. Precisely, the proof of \lemmaautorefname~7.3
in~\cite{CS-countpep} shows that if a $\PEPr$ instance admits a
solution $\sigma=i_1\ldots i_n$ longer than some bound $H$ then that
solution is not the shortest. Here $H$ is defined as
$2\cdot\bad_{(\Upgamma^\ast \cdot n)}(0)$ for a $n$ that is at most
exponential in the size of the instance. Since the control function is
linear, \autoref{theo-main} yields an $\FGH{\omega^{p-1}}$
complexity upper bound for $\PEPr$ on a $p$-letter alphabet (and a
hyper-Ackermannian $\FGH{\omega^\omega}$ when the alphabet is not
fixed).  This motivates a closer consideration of lower bounds (left
as future work, e.g., by adapting~\cite{CS-lics08}).

\ifthenelse{\boolean{proceedings_version}}{%
}{%
\paragraph{{\boldmath $F_{\omega^\omega}$}-Complete Problems}
Thanks to the lower bounds proved by \citet{CS-lics08} and our upper
bounds, reachability in lossy channel systems and $\PEPr$ are two
examples of problems complete for $F_{\omega^\omega}$.  This class
also includes several recently considered decision problems:
\begin{itemize}
\item model-checking for metric temporal logic~\citep{mtl},
\item universality for 1-clock timed automata~\citep{ata},  
\item emptiness for 1-register alternating automata over totally
  ordered data domains~\citep{FHL10}, and
\item reachability of finite concurrent programs with weak shared
  memory~\citep{wmreach}.
\end{itemize}
}

%% file: sec-concl.tex
\section{Concluding Remarks}\label{sec-concl}
We proved a general version of the Main Theorem promised in the
introduction. Our proof relies on two main components: an algebraic
framework for normed wqo's and normed reflections on the one hand,
leading on the other hand to descending relations between ordinals
that can be captured in subrecursive hierarchies. This setting
accommodates all ``exponential'' wqo's, i.e., finite combinations of
$\Upgamma_p^\ast$'s.  This lets us derive upper bounds for controlled bad
sequences when using Higman's Lemma on finite alphabets.

We hope that our framework will extend smoothly beyond exponential
wqo's and may also accept additional wqo constructions like powersets,
multisets, and perhaps trees.

%% file: appendix.tex
\newpage

\appendix
\pagestyle{myheadings}
\pagenumbering{roman}
\markboth{{Technical appendices.}}{{Technical appendices.}}

The following appendices provide the proofs missing from the main text
(Appendices~\ref{ax-nwqo} and~\ref{ax-otypes}), and further material
not required for the main developments: Appendix~\ref{app-comments}
provides some technical comments on the relationships with the
literature, and Appendix~\ref{ax-hierarchies} proposes the full proofs
of a few simple results from the literature we rely on, but for which
the proofs are not easily found in print (at least we do not know
where to find them).

\section{Proofs for Normed Wqo's and Reflections}
\label{ax-nwqo}

\subsection{Proof of \propositionautorefname~\ref{prop-L-finite}}
\label{app-prop-L-finite}

Since any prefix of a finite $n$-controlled bad sequence is $n$-controlled
and bad, these finite sequences ordered by the prefix ordering form a tree
$T$ with the empty sequence as its root. Now $T$ is finitely branching
since $\len{.}_A$ is proper. Furthermore, $T$ has no infinite branches
since $A$ is wqo. Hence, by K\H{o}nig's Lemma, there are only finitely many
branches in $T$, in particular finitely many maximal $n$-controlled bad sequences
over $A$, and a finite maximal length for them exists.

\subsection{Proof of \propositionautorefname~\ref{prop-descent-equation} (Descent Equation)}
\label{app-prop-descent-equation}

If $A_{<n}$ is empty, then $\bad_A(n)=0$ (the only controlled sequence
over $A$ is the empty sequence) while $\max\emptyset = 0$ by definition.
If $A_{<n}$ is not empty, we prove the two directions
of \eqref{eq-descent} independently.\medskip

\noindent
``$(\leq)$'': Write $\bad$ for $\bad_A(n)$ and let
$\xxx=x_0,x_1,x_2,\ldots,x_{\bad-1}$ be a maximal $n$-controlled bad
sequence over $A$. The suffix sequence $\xxx'=x_1,x_2,\ldots,x_{\bad-1}$ is
bad, is $g(n)$-controlled, and is over $A/x_0$: hence $\bad - 1\leq
\bad_{A/x_0}(g(n))$. Now $x_0\in A_{<n}$ (since $\xxx$ is controlled) and
we deduce one half of \eqref{eq-descent}.
\medskip

\noindent
``$(\geq)$'': Pick any $x\in A_{<n}$ and write $\bad'$ for
$\bad_{A/x}(g(n))$. This length is witnessed by a maximal $g(n)$-controlled
bad sequence, of the form $\xxx=x_1,\ldots,x_{\bad'}$. Since $\xxx$ is over
$A/x$, the sequence $\yyy\egdef x.\xxx$ over $A$, obtained by prefixing
$\xxx$ with $x$, is bad. It is also $n$-controlled. Hence $\bad_A(n)\geq
1+\bad'$, which concludes the proof.

\subsection{Proof of \propositionautorefname~\ref{prop-iso-wqos}}
\label{ax-iso-wqos}

These isomorphisms are classic for wqo's. For nwqo's, \emph{one must check
  that they preserve norms}.

We consider the main cases:
\begin{description}
\item[$\bm{1}\times A\equiv A$:] this
relies on
\[
\len{\tup{a_1,x}}_{\bm{1}\times A}
\equalby{\eqref{eq-len-prod}}
\max (\len{a_1}_{\bm{1}},\len{x}_A)
\equalby{\eqref{eq-lens}}
\max(0,\len{x}_A)
=
\len{x}_A		\: .
\]
\item[$\bm{0}^\ast\equiv\bm{1}$:] the only element
of $\bm{0}^\ast$ is $()$, the empty list. Norms are preserved:
\[
\len{()}_{\bm{0}^\ast}
\equalby{\eqref{eq-len-star}}
\max(0)
= 0
\:\equalby{\eqref{eq-lens}}\:
\len{a_1}_{\bm{1}}		      \: .
\]
\item[$\bm{1}^\ast\equiv\Nat$:] relies on an isomorphism that links the number $k$
in $\Nat$ with the unique length-$k$ list in $\bm{1}^\ast$. This preserves
norms:
\[
\len{(\obracew{a_1\ldots a_1}{k\text{ times}})}_{\bm{1}^\ast}
\equalby{\eqref{eq-len-star}}
\max(k,\len{a_1}_{\bm{1}},\ldots,\len{a_1}_{\bm{1}})
\equalby{\eqref{eq-lens}}
\max(k,0,\ldots,0)
= k
\equalby{\eqref{eq-lens}}
\len{k}_{\Nat}			  \: .
\]
\end{description}

\subsection{Proof of \propositionautorefname~\ref{prop-precongr}}
\label{app-prop-precongr}

Assuming $h:A\hookrightarrow A'$
and $h':B\hookrightarrow B'$, one immediately deduces
that $h+h':A+B\hookrightarrow A'+B'$, that $h\times h':A\times
B\hookrightarrow A'\times B'$, and that $h^\ast:A^\ast\hookrightarrow A'^\ast$
(assuming the obvious definitions for $h+h'$, $h\times h'$ and $h^\ast$).

Let us check, for example, that	 $h^\ast$ preserves non-comparability:
\begin{align*}
& h^\ast(x_1 \ldots x_n) \leq_{A'^\ast} h^\ast(y_1 \ldots y_m)
\\
\Leftrightarrow\;
& (h(x_1) \ldots h(x_n)) \leq_{A'^\ast} (h(y_1) \ldots h(y_m))
\tag{by def.\ of $h^\ast$}
\\
\Leftrightarrow\;
& h(x_1)\leq_{A'}h(y_{i_1}) \:\wedge\: \cdots \:\wedge\: h(x_n)\leq_{A'}h(y_{i_n})
\tag{by def.\ of $\leq_{A'^\ast}$}
\\
\shortintertext{for some $1\leq i_1<i_2<\cdots<i_n\leq m$,}
\Rightarrow\;
& x_1\leq_{A}y_{i_1} \:\wedge\: \cdots \:\wedge\:
x_n\leq_{A}y_{i_n}
\tag{since $h$ is a reflection}
\\
\Leftrightarrow\;
& (x_1 \ldots x_n)\leq_{A^*} (y_1 \ldots y_m)
\:.
\end{align*}
And that $h^\ast$ is norm-decreasing:
\begin{align*}
\len{h^\ast(x_1\ldots x_n)}_{A'^*}
& =
\len{(h(x_1)\ldots h(x_n))}_{A'^*}
\tag{by def.\ of $h^\ast$}
\\
& =
\max(n,\len{h(x_1)}_{A'},\ldots, \len{h(x_n)}_{A'})
\tag{by \eqref{eq-len-star}}
\\
& \leq
\max(n,\len{x_1}_{A},\ldots, \len{x_n}_{A})
\tag{since $h$ is norm-decreasing}
\\
& =
\len{(x_1\ldots x_n)}_{A^*}
\:.
\tag{by \eqref{eq-len-star}}
\end{align*}

\noindent
Using $\bm{0}\hookrightarrow A$ and $(A\equiv\bm{0})\vee(\bm{1}\hookrightarrow
A)$, we deduce for all $k\in\Nat$:
\begin{align}
\label{eq-mono-refl}
A\cdot k&\hookrightarrow A\cdot (k+1)
\:,
&
A^{k+1}&\hookrightarrow A^{k+2} \:.
\end{align}

\subsection{Proof of \propositionautorefname~\ref{prop-refl-residuals}}
\label{app-prop-refl-residuals}

The reflections in \eqref{eq-residual-sum} are in fact equalities and are obvious.
\medskip

\noindent
For
\eqref{eq-residual-prod}, the reflection relies on
\[
	\tup{x,y}\not\leq_{A\times B}\tup{x',y'}
\text{ iff }
	\bigl(x\not\leq_A x' \text{ or } y\not\leq_B y' \bigr)
\: ,
\]
which is just a rewording of \eqref{eq-order-prod}. It only provides a reflection, not an
isomorphism, because the ``or'' is not exclusive.
\medskip

For \eqref{eq-residual-star} and \eqref{eq-residual-finite-star},
we first observe the obvious equality
\begin{align}
\tag{$\dagger$}
\label{eq-residual-star1}
(A^\ast)/() &= \bm{0}			\: ,
\end{align}
that applies to empty lists. For non-empty lists, the following
lemma will be useful:
\begin{align}
\tag{$\star$}
\label{eq-residual-star2}
A^\ast/(x_1 \, x_2\ldots x_n) & \hookrightarrow
(A/x_1)^\ast\times \Bigl[\bm{1} + \uparrow_{\!A}x_1\times (A^\ast/(x_2\ldots x_n)) \Bigr]
\: ,
\end{align}
where $\uparrow_{\!A} x \:\eqdef\: \{y\in A\mid x\leq_A y\}$ denotes
the \emph{upward closure} of an element $x$ of $A$, seen as a
substructure of $A$.

\begin{proof}[Proof of \eqref{eq-residual-star2}]
We let $X\egdef (A/x_1)^\ast$, $Y\egdef
A^\ast/(x_2\ldots x_n)$, and exhibit a nwqo reflection to $X+X\times \uparrow_Ax_1\times
Y$, which is isomorphic to the target in \eqref{eq-residual-star2}. Write
$u$ for $(x_1\ldots x_n)$ and consider some $v=(y_1 \ldots y_m)\in A^\ast$.
Then $(x_1\ldots x_n)\leq_{A^\ast}(y_1\ldots y_m)$ is equivalent to
\begin{equation}
\label{eq-left-embed}
\tag{$\ddagger$}
\exists p\in\{1,\ldots,m\} \text{ s.t.\ }
\left\{
\begin{array}{l}
x_1\leq_A y_p
\\
(x_2\ldots x_n)\leq_{A^\ast}(y_{p+1}\ldots y_m)
\\
\forall 1\leq i<p:\; x_1\not\leq_A y_i
\end{array}
\right.%
\end{equation}
The third condition, ``$x_1\not\leq_A y_i$ for all $i<p$'', states that $p$ is the
leftmost position in $v$ where $x_i$ can be embedded. By negating
\eqref{eq-left-embed}, we see that $u\not\leq_{A^\ast} v$ iff there is no $p$
with $x_1\leq_A y_p$, or there is a leftmost one but $(x_2\ldots
x_n)\not\leq_{A^\ast}(y_{p+1}\ldots y_m)$. Therefore, any $v\in A^\ast/u$ (i.e.,
any $v$ s.t.\ $u\not\leq_{A^\ast}v$) is either a list in $(A/x_1)^\ast$, or can
be decomposed as a triple $\tup{(y_1\ldots y_{p-1}), y_p, (y_{p+1}\ldots y_m)}$
belonging to $(A/x_1)^\ast\times \uparrow_{\!A}x_1 \times (A^\ast/(x_2\ldots x_n))$. This
provides the required $h:A^\ast/u\rightarrow X+X\times \uparrow_{\!A}x_1\times Y$.

We now	check that $h$ is an order-reflection. For this assume that
$h(v)\leq_{X+X\times A\times Y} h(v')$. This requires that $v$ and
$v'$ are mapped to the same summand, and leads to two cases. If they
map to $X$, the left-hand summand, then $h(v)=v$, $h(v')=v'$, and
$h(v)\leq_{(A/x_1)^\ast}h(v')$, i.e., $h(v)\leq_{(A/x_1)^\ast}h(v')$,
implies $v\leq_{A^\ast} v'$. If they map to $X\times
\uparrow_{\!A}x_1\times Y$, then $h(v)$ is some $\tup{v_1,y,v_2}$,
$h(v')$ is some $\tup{v'_1,y',v'_2}$, and $h(v)\leq h(v')$ implies
$v_1\leq_X v'_1$, $y\leq_{\uparrow_{\!A}x_1} y'$, and $v_2\leq_Y
v'_2$. Now, since $X$, $Y$, and $\uparrow_{\!A}x_1$, are substructures
of $A^\ast$ and $A$, we deduce $v_1\leq_{A^\ast}v'_1$,
$v_2\leq_{A^\ast}v'_2$, and $y\leq_A y'$. Since $v$ and $v'$ are
exactly $v_1.y.v_2$ and, respectively, $v'_1.y'.v'_2$, we deduce
$v\leq_{A^\ast} v'$ from the compatibility of embedding with
concatenation.

Finally, we let the reader check that $h$ is norm-decreasing, and
observe that the reason why Eq.~\eqref{eq-residual-star2} is not an
isomorphism is because the norms are not preserved in the
decomposition $v\mapsto\tup{v_1,y,v_2}$.
\end{proof}
\begin{proof}[Proof of \eqref{eq-residual-star}]
We now prove \eqref{eq-residual-star} by induction on $n$. The base
case, $n=0$, is provided by \eqref{eq-residual-star1} since
$\Upgamma_0\times A^0\equiv\bm{0}$.
For the inductive case we assume $n>0$, which implies $A\not=\bm{0}$.
By ind.\ hyp., $A^\ast/(x_2\ldots x_n)\hookrightarrow
\Upgamma_{n-1}\times A^{n-1}\times\prod_{i=2}^n (A/x_i)^\ast$.
Replacing in \eqref{eq-residual-star2}, we deduce
\begin{align*}
A^\ast/(x_1\ldots x_n)
& \,\hookrightarrow\, (A/x_1)^\ast\times\Bigl[\bm{1}
\:+\:
\uparrow_{\!A}x_1\times
\Upgamma_{n-1}\times A^{n-1}\times\prod_{i=2}^n(A/x_i)^\ast
\Bigr]
\\
& \,\hookrightarrow\, (A/x_1)^\ast \times
\Bigl[\bm{1} \:+\: \Upgamma_{n-1}\times A^{n}\times\prod_{i=2}^n(A/x_i)^\ast\Bigr]
\tag{by $\uparrow_{\!A}x\hookrightarrow A$}
\\
\shortintertext{and since
$\bm{1}\hookrightarrow A^{n}\times\prod_{i=2}^n(A/x_i)^\ast$ (recall that $A\not=\bm{0}$):}
& \,\hookrightarrow\, (A/x_1)^\ast \times
\Bigl[(\bm{1}+\Upgamma_{n-1})\times A^{n}\times\prod_{i=2}^n(A/x_i)^\ast\Bigr]
\\
& \,\equiv\, \Upgamma_n\times A^{n}\times\prod_{i=1}^n(A/x_i)^\ast
\: .\qedhere
\end{align*}
\end{proof}

\begin{proof}[Proof of \eqref{eq-residual-finite-star}]
By induction on $n$. The base case, $n=0$, is provided by
Eq.~\eqref{eq-residual-star1} since
$\Upgamma_0\times(\Upgamma_p^\ast)^0\equiv\bm{0}$.

For the inductive case, $n>0$, we first simplify
\eqref{eq-residual-star2} using $\Upgamma_{p+1}/x \equiv \Upgamma_{p}$
from Eq.~\eqref{eq-ex-residuals}, and noting that
$\uparrow_{\Upgamma_{p+1}}x \equiv \bm{1}$ since it amounts to the
singleton $\{x\}$. This yields
\begin{align*}
    \Upgamma_{p+1}^\ast/(x_1 \, x_2\ldots x_n) & \hookrightarrow
    \Upgamma_{p}^\ast\times\bigl(\bm{1}+\Upgamma_{p+1}^\ast/(x_2\ldots x_n)\bigr)
    \: .
\shortintertext{Using the ind.\ hyp., $\Upgamma_{p+1}^\ast/(x_2\ldots x_n)\hookrightarrow \Upgamma_{n-1}\times(\Upgamma_p^\ast)^{n-1}$, one obtains}
\Upgamma_{p+1}^\ast/(x_1\ldots x_n)
&\hookrightarrow\Upgamma_{p}^\ast\times\bigl(\bm{1} + \Upgamma_{n-1}\times(\Upgamma_{p}^\ast)^{n-1}
\bigr)\\
&\:\equiv\,\Upgamma_{p}^\ast+\Upgamma_{n-1}\times(\Upgamma_{p}^\ast)^{n}\\
&\hookrightarrow (\Upgamma_{p}^\ast)^n + \Upgamma_{n-1}\times(\Upgamma_{p}^\ast)^n
\tag{by Eq.~\eqref{eq-mono-refl}}
\\
&\:\equiv\,\Upgamma_{n}\times(\Upgamma_{p}^\ast)^n\:.\qedhere
\end{align*}
\end{proof}

\section{Proofs for Ordinals and Subrecursive Hierarchies}
\label{ax-otypes}

\subsection{Well-Foundedness of Derivatives}
\label{app-deriv-is-wf}

This requires basic monotonicity properties of natural sums and
products:
\begin{align}
\label{eq-mono-nat-sum}
\alpha<\alpha' &\text{ implies }
\alpha\oplus\beta < \alpha'\oplus\beta
\: ,
\\
\label{eq-mono-nat-prod}
\alpha<\alpha' \:\wedge\: 0<\beta &\text{ implies }
\alpha\otimes \beta<\alpha'\otimes\beta
\: ,
\\
\shortintertext{and a direct consequence of Eq.~\eqref{eq-def-ordering}, the defining property of principal ordinals:}
\label{eq-defn-principals}
\Bigl(\bigoplus_{i=1}^n\alpha_i\Bigr)<\omega^\beta
&\text{ iff }
\alpha_1<\omega^\beta
\:\wedge\:
\cdots
\:\wedge\:
\alpha_n<\omega^\beta
\:.
\end{align}

We can now prove that $\alpha'\in\partial_n \alpha$ implies
$\alpha'<\alpha$.

One first checks that $D_n(\alpha)<\alpha$ for all principal ordinals. This
is immediate in the case of multiplicative principals: see Eq.~\eqref{eq-dn}.
For the more general case $\alpha=\bigotimes_i\omega^{\omega^{p_i}}$,
Eq.~\eqref{eq-Dn} gives $D_n(\alpha)$ as a sum of terms that are individually
smaller than $\alpha$ thanks to Eq.~\eqref{eq-mono-nat-prod}. One
concludes with Eq.~\eqref{eq-defn-principals}.

Finally, knowing that $D_n(\omega^\beta)<\omega^\beta$,
Eq.~\eqref{eq-partialn} combined with Eq.~\eqref{eq-mono-nat-sum}
proves $\alpha'<\alpha$ when $\alpha'\in\partial_n\alpha$.

\subsection{Proof of \theoremautorefname~\ref{theo-deriv-refl}}
\label{app-theo-deriv-refl}

We want to prove that, for $x\in A_{<n}$, $A/x$ can be reflected in
$C(\alpha')$ for some $\alpha'\in\partial_n o(A)$.

We write $A$ in canonical form $A\equiv
\sum_{i=1}^m\prod_{j=1}^{k_i}\Upgamma^\ast_{p_{i,j}+1}$, as \textit{per}
Eq.~\eqref{eq-Cdef}, and consider special cases first:
\begin{description}
\item[Case 1, $A$ is finite (i.e., $k_i=0$ for all $i$):]
Then $A\equiv\sum_{i=1}^m\bm{1}\equiv\Upgamma_m$. If $m=0$, i.e., $A\equiv\bm{0}$, the
claim holds vacuously since there is no $x\in A$. If $m>0$ then
$o(A)\equalby{\eqref{eq-odef-1}}m$, $\partial_n m
\equalby{\eqref{eq-deriv-2}} \{m-1\}$, $C(m-1)\equalby{\eqref{eq-Cdef}}\Upgamma_{m-1}$ and we
conclude with Eq.~\eqref{eq-ex-residuals}: $\Upgamma_m/x\equiv\Upgamma_{m-1}$.

\item[Case 2, $A$ is some $\Upgamma_{p+1}^\ast$ (i.e., $m=1=k_1$):]

Then $o(A)\equalby{\eqref{eq-odef-1}}\omega^{\omega^{p}}$. By
Eq.~\eqref{eq-dn}, $\partial_n o(A)$ gives a single $\alpha'$ that is $n-1$
if $p=0$, and $\omega^{(\omega^{p-1}\cdot(n-1))}\cdot(n-1)$ if $p>0$. In the first case,
$C(\alpha')=\Upgamma_{n-1}$. In the second case
\[
C\Bigl(\omega^{(\omega^{p-1}\cdot(n-1))}\cdot(n-1)\Bigr) \equalby{\eqref{eq-Cdef}}
(\Upgamma_p^\ast)^{n-1}\cdot(n-1) \equiv (\Upgamma_{p}^\ast)^{n-1}\times
\Upgamma_{n-1}\:.
\]
Thus, in view of $\Upgamma_0^\ast\equiv\bm{1}$, we can write
$C(\alpha')\equiv (\Upgamma_{p}^\ast)^{n-1}\times \Upgamma_{n-1}$ even
when $p=0$. On the other hand, Eq.~\eqref{eq-residual-finite-star}
yields $A/x\hookrightarrow A'\egdef
(\Upgamma_{p}^\ast)^{\len{x}}\times\Upgamma_{\len{x}}$. But since
$\len{x}_A < n$, one deduces $A'\hookrightarrow C(\alpha')$ from the
algebraic properties of reflection.

\item[Case 3, $A$ is a product $\prod_{i=1}^k\Upgamma_{p_i+1}^\ast$, i.e.,
  $m=1\leq k_1$:]
Then
\begin{align}
o(A)\equalby{\eqref{eq-odef-2}} \bigotimes_{i=1}^k o(\Upgamma_{p_i+1}^\ast)
\equalby{\eqref{eq-odef-1}} \bigotimes_{i=1}^k \omega^{\omega^{p_i}}
&= \omega^{(\omega^{p_1}\oplus\cdots\oplus\omega^{p_k})}
\: .
\\
\shortintertext{Hence $\partial_n o(A)$ gives a single $\alpha'=D_n(o(A))$ and}
\label{eq-aux2}
C(\alpha')
\equalby{\eqref{eq-Dn}}
C\Bigl[\bigoplus_{i=1}^k\Bigl(\:\obracew{D_n\bigl(\omega^{\omega^{p_i}}\bigr)}{\beta_i=}\otimes\obracew{\bigotimes_{\ell\neq i}\omega^{\omega^{p_\ell}}}{\alpha'_i=}\:\Bigr)
\Bigr]
&\equiv \sum_{i=1}^k C(\beta_i)\times C(\alpha'_i)
\: .
\end{align}
Now $C(\alpha'_i)\equiv A_i\egdef\prod_{\ell\not=i}\Upgamma^\ast_{p_\ell+1}$
since $C$ is the inverse of $o$.
On the other hand, $x\in A$ must have the form $\tup{x_1,\ldots,x_k}$. With
Eq.~\eqref{eq-residual-prod} we see that
\begin{equation}
\label{eq-aux3}
A/x\hookrightarrow
\sum_{i=1}^k
\Bigl((\Upgamma_{p_i+1}^\ast)/x_i\times\prod_{l\not=i}\Upgamma_{p_l+1}^\ast\Bigr)
=\sum_{i=1}^k
\Bigl((\Upgamma_{p_i+1}^\ast)/x_i\times C(\alpha'_i)\Bigr)
\:.
\end{equation}
We saw (Case 2) that $\Upgamma^\ast_{p_i+1}/x_i\hookrightarrow C(\beta_i)$.
Combining with Eq.~\eqref{eq-aux2} and~\eqref{eq-aux3}, we conclude that
$A/x\hookrightarrow C(\alpha')$.

\item[Case 4, $A$ is $\sum_{i=1}^m\prod_{j=1}^{k_i}\Upgamma_{p_{i,j}+1}^\ast$ with $m>1$:]
We write $A=\sum_{i=1}^m A_i$, so that $o(A)=\bigoplus_{i=1}^m o(A_i)$ and
\[
\partial_n o(A) = \Bigl\{ D_n(o(A_i))\oplus \bigoplus_{\ell\not=i}o(A_\ell)
~\Bigl\lvert~ i=1,\ldots,m\Bigr\}	    \:.
\]
On the other hand, we know that $x$ is $\tup{i,x'}$ for some
$i\in\{1,\ldots,m\}$. With Eq.~\eqref{eq-residual-sum}, we deduce that
$A/x\hookrightarrow A_i/x' + \sum_{\ell\not=i}A_\ell$. By picking
$\alpha'=D_n(o(A_i)) \oplus \bigoplus_{\ell\not=i}o(A_\ell)$, we
deduce $A_i/x' + \sum_{\ell\not=i}A_l\hookrightarrow C(\alpha')$ since
$A_\ell \equiv C(o(A_\ell))$ and $A_i/x'\hookrightarrow D_n(o(A_i))$
as we saw with Case~3.
\end{description}

\subsection{Lean Ordinals and Pointwise Ordering}

We present some intermediate results before we can prove
\autoref{prop-hardy}.

A key issue with hierarchies like $(h_\alpha)_{\alpha<\ezero}$ and
$(f_\alpha)_{\alpha<\ezero}$ is that, in general, $\alpha<\alpha'$,
does not imply $h_\alpha(x)\leq h_{\alpha'}(x)$ or $f_\alpha(x)\leq
f_{\alpha'}(x)$. Such monotonicity w.r.t.\ $\alpha$ only holds
``eventually'', or for ``sufficiently large $x$''. This
issue appears very quickly since just proving monotonicity in the $x$
argument \emph{requires} some
monotonicity in the $\alpha$ index in the case where $\alpha$ is a
limit.

Indeed, we will use some of that monotonicity w.r.t.\ $\alpha$ in
order to handle the
``$\max_{\alpha'\in\partial_n\alpha}\ubound_{\alpha'}(\ldots)$'' in
\eqref{eq:ubound} and majorize it by some
``$\ubound_{\max(\partial_n\alpha)}(\ldots)$''.

\paragraph{A Refined Ordering.}
In order to deal with these issues, a standard solution goes through a
ternary relation between $x$, $\alpha$ and $\alpha'$, as we now
explain.

For each $x$ in $\Nat$, define a relation $\dd$ between ordinals,
called ``\emph{pointwise-at-$x$ ordering}''
in~\cite{cichon98}, as the smallest transitive relation s.t.\ for all
$\alpha$, $\lambda$:
\begin{xalignat}{2}
  \label{eq-dd-def}
			\alpha&\dd\alpha+1\;,
  &
		      \lambda_{x}&\dd\lambda\;.
\end{xalignat}

The inductive definition of $\dd$ implies
\begin{equation}
\label{eq-dd-inductive}
\alpha'\dd\alpha
\text{ iff }
\left\{
\begin{array}{l}
\text{$\alpha = \beta+1$ is a successor and $\alpha'\dde\beta$, or}
\\
\text{$\alpha = \lambda$ is a limit and $\alpha'\dde\lambda_x$.}
\end{array}\right.%
\end{equation}

Obviously $\dd$ is a restriction of $<$, the linear ordering of
ordinals. For example, $x+1=\omega_x\dd\omega$ but $x+2\not\dd\omega$.
The $\dd$ relations are linearly ordered themselves, and $<$, can be
recovered in view of:
\begin{equation}\label{eq-hierarchy-dd}
   \dd[0] \;\subset\; \cdots \;\subset\;\dd[x]\; \subset \; \dd[x+1]
   \;\subset\; \cdots \;\subset \Bigl(\bigcup_{x\in\Nat}\dd\Bigr) = \; <
\:.
\end{equation}
\noindent
More precisely, we prove in \autoref{ssec-ax-P-y-dd} the following
results when $\omega_x=x+1$, from which the inclusions in
\eqref{eq-hierarchy-dd} follow:
\begin{align}
\label{eq-dd-zero}
  &0 \dde\alpha
\;,
\\
\label{eq-dd-sum}
 \alpha'\dd\alpha \text{ implies }  &\gamma+\alpha'\dd\gamma+\alpha
\;,
\\
\label{eq-dd-power}
 \alpha'\dd\alpha \text{ implies }  &\omega^{\alpha'}\dd\omega^{\alpha}
\;,
\\
\label{eq-dd-limit}
 x<y \text{ implies }  &\lambda_x\dd[y]\lambda_y
\;.
\end{align}
With this, one can show \citep[see][\theoremautorefname~2, or Appendix~\ref{sub-pointwise}]{cichon98} that, for smooth $h$:
\begin{align}
  \label{eq-mono-hardy}
  x<y&\text{ implies }h_\alpha(x)\leq h_\alpha(y)\;,\\
  \label{eq-dd-hardy}
  \alpha'\dd[x]\alpha&\text{ implies }h_{\alpha'}(x)\leq
  h_{\alpha}(x)\;.
\end{align}

\paragraph{Lean Ordinals.}
Now, in order to use Eq.~\eqref{eq-dd-hardy} in the analysis of
$\ubound$, we need to show that $\alpha'\dd\alpha$ when
$\alpha'\in\partial_n\alpha$. This is handled through a notion of
\emph{lean} ordinals, as we now explain.

Let $k$ be in $\Nat$. We say that an ordinal $\alpha$ in
$\CNF[\ezero]$ is \emph{$k$-lean} if it only uses coefficients $\leq
k$, or, more formally, when it is written under the strict form
$\alpha=\omega^{\beta_1}\cdot c_1+\cdots+\omega^{\beta_m}\cdot c_m$
with $c_i\leq k$ and, inductively, with $k$-lean $\beta_i$, this for
all $i=1,...,m$.
Observe that only $0$ is $0$-lean, and that if $\alpha$ is $k$-lean
and $\alpha'$ is $k'$-lean, then $\alpha\oplus\alpha'$ is
$(k+k')$-lean.
\medskip

\noindent
Leanness is a fundamental tool when it comes to understanding the
$\dd$ relation:
\begin{lemma}[see \autoref{ssec-ax-P-y-dd}]
\label{lem-Px-dd}
Let  $\alpha$ be $x$-lean. Then 
\begin{equation}
\label{eq-Px-dd}
\alpha<\gamma
\;\;\text{iff}\;\;
\alpha\dde P_x(\gamma)
\;\Bigl[\:\text{also: iff}\;\;
\alpha\dd\gamma
\;\;\text{iff}\;\;
\alpha\leq P_x(\gamma)
\:\Bigr]
\:.
\end{equation}
\end{lemma}

\subsection{Bounding $\ubound$: Proof of \propositionautorefname~\ref{prop-hardy}}
\label{ax-hardy}

We first bound the leanness of derivatives
$\alpha'\in\partial_n\alpha$ in function of $n$ and $\alpha$.
\begin{proposition}
\label{coro-partial-lean}
Assume $k,n>0$ and $\alpha\in\CNF$ is $k$-lean. If
$\alpha\mathrel{\partial_n}\alpha'$ then $\alpha'$ is $kn$-lean.
\end{proposition}
\begin{proof}
  We first show that $D_n(\omega^\beta)$ is $k(n-1)$-lean
  when $\beta\in\CNF[\omega^\omega]$ is $k$-lean.  For this, write
  $\beta$ under the strict form $\beta=\sum_{i=1}^m\omega^{p_i} \cdot
  c_i$.  Now Eq.~\eqref{eq-Dn} gives:
\begin{align}
\notag
    D_n(\omega^\beta)
    & =\bigoplus_{i=1}^m\Bigl[D_n(\omega^{\omega^{p_i}})\cdot c_i\otimes\omega^{\omega^{p_i}\cdot (c_i-1)}\otimes\bigotimes_{\ell\neq i}\omega^{\omega^{p_\ell}\cdot c_\ell}\Bigr]
\\
\label{eq-new-aux2}
    & = \bigoplus_{i=1}^m \bigl[ \omega^{\beta_i}\cdot c_i(n-1) \bigr]
\:,
\\
\shortintertext{with}
\label{eq-def-beta-i}
  \beta_i
    &\egdef \omega^{p_i}\cdot(c_i-1) \oplus \ubracew{\omega^{p_i-1}\cdot (n-1)}{\text{or $0$ if $p_i=0$}} \oplus \bigoplus_{\ell\not= i}\omega^{p_\ell}\cdot c_\ell
\:.
  \end{align}
  We can assume $n>1$ since otherwise $D_n(\omega^\beta)=0$ and we are
  done. Inspecting Eq.~\eqref{eq-def-beta-i}, we see that the coefficients
  in $\beta_i$ are $c_i-1$, $n-1$, $c_\ell$ for $\ell\not=i$, and can
  even be $(n-1)+c_{i+1}$ in the case where $p_i-p_{i+1}=1$. Since
  $\beta$ is $k$-lean,  these coefficients
  are $\leq k+n-1\leq k(n-1)$.  Hence $\beta_i$ is
  $k(n-1)$-lean.  The same holds of $D_n(\omega^\beta)$
  since all the $c_i(n-1)$'s in Eq.~\eqref{eq-new-aux2} are $\leq
  k(n-1)$, and cannot get combined in view of
  $\beta_m>\beta_{m-1}>\cdots>\beta_1$.
  \medskip

  Now assume $\alpha'\in\partial_n\alpha$ and write $\alpha$ under the
  form $\alpha=\gamma\oplus\omega^\beta$ such that
  $\alpha'=\gamma\oplus D_n(\omega^\beta)$. We just proved that
  $D_n(\omega^\beta)$ is $k(n-1)$-lean, and since $\gamma$ is
  $k$-lean, $\alpha'$ is $(k(n-1)+k)$-lean, i.e., $kn$-lean.
\end{proof}

\begin{corollary}\label{cor-hardy}
  Let $k,n>0$, $\alpha,\alpha'$ be in $\CNF$, and $h$ be a smooth
  function.  If $\alpha$ is $k$-lean and $\alpha'\in\partial_n\alpha$,
  then for all $x\geq kn$,
  \begin{equation*}
	    h_{\alpha'}(x) \leq h_{P_{kn}(\alpha)}(x)
\;.
  \end{equation*}
\end{corollary}
\begin{proof}
  Since $\alpha$ is $k$-lean, then $\alpha'$ is $kn$-lean by
  \autoref{coro-partial-lean}, hence $\alpha'\dde[kn]P_{kn}(\alpha)$
  by \autoref{lem-Px-dd} and thus $\alpha'\dde P_{kn}(\alpha)$ by
  \eqref{eq-hierarchy-dd}. One concludes by Eq.~\eqref{eq-dd-hardy}.
\end{proof}

We can now bound $\ubound_\alpha(n)$. Let $h(x)\egdef x\cdot g(x)$ and
note that $h$ is smooth since $g$ is. The following claim proves
\autoref{prop-hardy}, e.g., by choosing $k$ as the leanness of
$\alpha$.
\begin{claim}
  Let $n>0$ and $\alpha$ in $\CNF$.  If $\alpha$ is $k$-lean, then
  \begin{equation*}
    \uboundd{\alpha}(n)\leq h_{\alpha}(kn)\;.
  \end{equation*}
\end{claim}
\begin{proof}
  By induction on $\alpha$.  If $\alpha=0$, then
  $\uboundd{0}(n)=0=h_{0}(kn)$.  Otherwise, $k>0$ and there exists
  $\alpha'$ in $\partial_{n}\alpha$ s.t.\
  $\uboundd{\alpha}(n)=1+\uboundd{\alpha'}(g(n))$.  Observe that
  $\alpha'<\alpha$, and that by \autoref{coro-partial-lean}, $\alpha'$
  is $kn$-lean, which means that we can use the induction hypothesis:
  \begin{align*}
    \uboundd{\alpha}(n)&=1+\uboundd{\alpha'}(g(n))\\
    &\leq 1+h_{\alpha'}\!\left(kn\cdot g(n)\right)\tag{by ind.\
      hyp.}\\
    &\leq 1+h_{\alpha'}\!\left(kn\cdot g(kn)\right)\tag{by
      monotonicity of $g$ and $h_{\alpha'}$, since $k>0$}\\
    &=1+h_{\alpha'}\!\left(h(kn)\right)\tag{by def.\ of $h$}\\
    &\leq 1+h_{P_{kn}(\alpha)}\!\left(h(kn)\right)\tag{by
      \autoref{cor-hardy} since $h(x)\geq x$}\\
    &=h_{\alpha}(kn)\;.\tag{by \autoref{lem-hierarchies}}
  \end{align*}
\end{proof}
Note that a close inspection of the proof of
\autoref{coro-partial-lean} would allow to refine this bound to
\begin{equation*}
  \ubound_\alpha(n)\leq h_\alpha\!\left((k-1)n-k+1\right)
\end{equation*}
at the expense of readability.  See \autoref{ax:comps} for detailed
comparisons with other bounds in the literature.

\subsection{Proof of \theoremautorefname~\ref{theo-main}}

\begin{proof}[Proof sketch for \autoref{theo-main}]
  Observe that $h(x)\eqdef x\cdot g(x)$ is in $\FGH{\max(\gamma,2)}$.
  We apply \autoref{prop-hardy}, items 2--3 of
  \autoref{lem-hierarchies}, and \autoref{lem-fgh} to show that
  $\badd{A,g}$ is bounded above by a function in
  \begin{itemize}
  \item $\FGH{\beta}$ if $\gamma<\omega$ and $\beta\geq\omega$, and in
  \item $\FGH{\max(\gamma,2)+\beta}$ if $\beta<\omega$.
  \end{itemize}
  See \autoref{sub:ext-grz} for details on the $(\FGH{\alpha})_\alpha$
  hierarchy.
\end{proof}

%% file: hierarchies.tex
\section{Subrecursive Hierarchies}
\label{ax-hierarchies}
In a few occasions (namely \autoref{lem-hierarchies},
\eqref{eq-mono-hardy}, and \eqref{eq-dd-hardy}), we refer to results
stated \emph{without proof} by \citet{cichon83} and \citet{cichon98}.
As we do not know where to find these proofs (they are certainly too
trivial to warrant being published in full), we provide these missing
proofs in this appendix, which might still be helpful for readers
unaccustomed to subrecursive hierarchies.  The appendix is also the
occasion of checking that minor variations we have made in the
definitions of the hierarchies are harmless, and of proving useful
results on lean ordinal terms.

\subsection{Ordinal Terms}
\renewcommand{\CNF}[1][\ezero]{\ensuremath\mathrm{CNF}(#1)}
We work as is most customary with the set $\Omega$ of \emph{ordinal
  terms} following the abstract syntax
\begin{equation*}
  \alpha::= 0\mid \omega^\alpha\mid \alpha+\alpha
\end{equation*}
for ordinals below $\ezero$.  We write $1$ for $\omega^0$ and
$\alpha\cdot n$ for $\obracew{\alpha+\cdots+\alpha}{n\text{ times}}$.
We work modulo associativity
($(\alpha+\beta)+\gamma=\alpha+(\beta+\gamma)$) and idempotence
($\alpha+0=\alpha=0+\alpha$) of $+$.  An ordinal term $\alpha$ of form
$\gamma+1$ is called a \emph{successor ordinal}.  Otherwise, if not $0$, it is a
\emph{limit ordinal}, usually denoted $\lambda$.

Each ordinal term $\alpha$ denotes a unique ordinal
$\mathop{ord}(\alpha)$ (by interpretation into ordinal arithmetic,
with $+$ denoting direct sum),
from which we can define a well-ordering on terms by
$\alpha'\leq\alpha$ if
$\mathop{ord}(\alpha')\leq\mathop{ord}(\alpha)$. Note that the mapping
of terms to ordinals is not injective, so the
ordering on terms is not antisymmetric.

Ordinal terms can be in Cantor Normal Form (CNF), i.e.\ sums
\begin{equation*}
  \alpha = \omega^{\beta_1}+\cdots+\omega^{\beta_m}
\end{equation*}
with $\alpha>\beta_1\geq\cdots\geq\beta_m\geq 0$ with each $\beta_i$
in CNF itself.  We also use at times the \emph{strict} form
\begin{equation*}
  \alpha = \omega^{\beta_1}\cdot c_1+\cdots+\omega^{\beta_m}\cdot c_m
\end{equation*}
where $\alpha>\beta_1>\cdots>\beta_m\geq 0$ and
$\omega>c_1,\ldots,c_m>0$ and each $\beta_i$ in strict form---we call
the $c_i$'s \emph{coefficients}.  Terms $\alpha$ in CNF are in
bijection with their denoted ordinals $\mathop{ord}(\alpha)$.  We write
$\CNF[\alpha]$ for the set of ordinal terms $\alpha'<\alpha$ in CNF;
thus $\CNF$ is a subset of $\Omega$ in this view.

When working with terms in $\CNF$, the ordering has a syntactic
characterization as
\begin{equation}
\alpha < \alpha'
\Leftrightarrow
\left\{\begin{array}{l}
\text{$\alpha=0$ and $\alpha'\not=0$, or}
\\
\text{$\alpha=\omega^\beta+\gamma$, $\alpha'=\omega^{\beta'}+\gamma'$
  and}
\left\{\begin{array}{l}
\text{$\beta<\beta'$, or}
\\
\text{$\beta=\beta'$ and $\gamma<\gamma'$.}
\end{array}\right.%
\end{array}\right.%
\tag{see \eqref{eq-def-ordering}}
\end{equation}

\subsection{Predecessors and Pointwise Ordering}
\label{ssec-ax-P-y-dd}

\paragraph{Fundamental Sequences.}
Subrecursive hierarchies are defined through assignments of
\emph{fundamental sequences} $(\lambda_x)_{x<\omega}$ for limit
ordinal terms $\lambda$ in $\Omega$, verifying $\lambda_x<\lambda$ for
all $x$ and $\lambda=\sup_x\lambda_x$.  One way to obtain families of
fundamental sequences is to fix a particular sequence $\omega_x$ for
$\omega$ and to define
\begin{align}\tag{see \ref{eq-fund}}
  (\gamma+\omega^{\beta+1})_{x}&\eqdef\gamma+\omega^\beta\cdot\omega_x\,,&
  (\gamma+\omega^\lambda)_{x}&\eqdef\gamma+\omega^{\lambda_x}\;.
\end{align}
We assume $\omega_x$ to be the value in $x$ of some monotone function
$s$ s.t.\ $s(x)\geq x$ for all $x$, typically $s(x)=x$ or $s(x)=x+1$
(as in the main text).  We will see in \autoref{sub:fund-seq} how
different assignments of fundamental sequences for $\omega$ influence
the hierarchies of functions built from them.

\paragraph{Predecessors.}

Given an assignment of fundamental sequences, one defines the
($x$-indexed) \emph{predecessor} $P_x(\alpha)<\alpha$ of an ordinal
$\alpha\not=0$ in $\Omega$ as
\begin{align}\tag{see \ref{eq-pred}}
 P_x(\alpha+1)&\eqdef\alpha\,,&P_x(\lambda)&\eqdef P_x(\lambda_x)\;.
\end{align}

\renewcommand{\CNF}[1][\ezero]{\ensuremath\mathrm{CNF}(#1)}

\begin{lemma}
Assume $\alpha>0$ in $\Omega$. Then for all $x$
\begin{align}
\label{eq-Px-prop1}
P_x(\gamma+\alpha) &= \gamma+P_x(\alpha)
\: ,
\\
\label{eq-Px-prop2}
\text{if }\omega_x>0\text{ then }P_x(\omega^\alpha) &= \omega^{P_x(\alpha)}\cdot(\omega_x-1)+P_x(\omega^{P_x(\alpha)})
\: .
\end{align}
\end{lemma}
\begin{proof}[Proof of \eqref{eq-Px-prop1}]
By induction over $\alpha$.
For the successor case $\alpha=\beta+1$, this goes
\begin{align*}
&P_x(\gamma+\beta+1)
\equalby{\eqref{eq-pred}}
\gamma+\beta
\equalby{\eqref{eq-pred}}
\gamma+P_x(\beta+1)\:.
\\
\shortintertext{For the limit case $\alpha=\lambda$, this goes
}
&P_x(\gamma+\lambda)
\equalby{\eqref{eq-pred}}
P_x((\gamma+\lambda)_x)
\equalby{\eqref{eq-fund}}
P_x(\gamma+\lambda_x)
\equalby{ih}
\gamma+P_x(\lambda_x)
\equalby{\eqref{eq-pred}}
\gamma+P_x(\lambda)\;.\qedhere
\end{align*}
\end{proof}
\begin{proof}[Proof of \eqref{eq-Px-prop2}]
By induction over $\alpha$.
For the successor case $\alpha=\beta+1$, this goes
\begin{align*}
\!\!\!P_x(\omega^{\beta+1})
&\equalby{\eqref{eq-pred}}
P_x((\omega^{\beta+1})_x)
\equalby{\eqref{eq-fund}}
P_x(\omega^{\beta}\cdot \omega_x)
\equalby{\eqref{eq-Px-prop1}}
\omega^{\beta}\cdot(\omega_x-1)+
P_x(\omega^{\beta})
\\
&\equalby{\eqref{eq-pred}}
\omega^{P_x(\beta+1)}\cdot(\omega_x-1)+
P_x(\omega^{P_x(\beta+1)})
\:.
\\
\shortintertext{For the limit case $\alpha=\lambda$, this goes
}
P_x(\omega^{\lambda})
&\equalby{\eqref{eq-pred}}
P_x((\omega^{\lambda})_x)
\equalby{\eqref{eq-fund}}
P_x(\omega^{\lambda_x})
\equalby{ih}
\omega^{P_x(\lambda_x)}\cdot (\omega_x-1)+
P_x(\omega^{P_x(\lambda_x)})
\\
&\equalby{\eqref{eq-pred}}
\omega^{P_x(\lambda)}\cdot (\omega_x-1)+
P_x(\omega^{P_x(\lambda)})\:.\qedhere
\end{align*}
\end{proof}

\paragraph{Pointwise ordering.}
Recall that, for $x\in\Nat$, $\dd$ is the smallest transitive relation
satisfying:
\begin{xalignat}{2}
  \tag{see \ref{eq-dd-def}}
			\alpha&\dd\alpha+1\;,
  &
		      \lambda_{x}&\dd\lambda\;.
\end{xalignat}
In particular, using induction on $\alpha$, one immediately sees that
\begin{align}
\tag{see \ref{eq-dd-zero}}
			    0 &\dde\alpha
\;,
\\
		       P_{x}(\alpha)&\dd\alpha
\:.
\end{align}

\begin{lemma}
For all $\alpha$ in $\Omega$ and all $x$
\begin{align*}
\tag{see \ref{eq-dd-sum}}
 \alpha'\dd\alpha \text{ implies }  &\gamma+\alpha'\dd\gamma+\alpha
\:,
\\
\tag{see \ref{eq-dd-power}}
 \omega_x>0\text{ and }\alpha'\dd\alpha \text{ imply } &\omega^{\alpha'}\dd\omega^\alpha
\:.
\end{align*}
\end{lemma}
\begin{proof}
All proofs are by induction over $\alpha$ (NB: the case
$\alpha=0$ is impossible).
\\
\textbf{\eqref{eq-dd-sum}:}
For the successor case $\alpha=\beta+1$, this goes through
\begin{align*}
\alpha'\dd\beta+1
\text{ implies }
&\alpha'\dde\beta
\tag{by Eq.~\eqref{eq-dd-inductive}}
\\
\text{ implies }
&\gamma+\alpha'\dde\gamma+\beta
\ddby{\eqref{eq-dd-inductive}}\gamma+\beta+1
\:.
\tag{by ind.\ hyp.}
\\
\shortintertext{For the limit case $\alpha=\lambda$, this goes through
}
\alpha'\dd\lambda
\text{ implies }
&\alpha'\dde\lambda_x
\tag{by Eq.~\eqref{eq-dd-inductive}}
\\
\text{ implies }
&\gamma+\alpha'\dde\gamma+\lambda_x\equalby{\eqref{eq-fund}}(\gamma+\lambda)_x
\ddby{\eqref{eq-dd-def}}\gamma+\lambda
\: .
\tag{by ind.\ hyp.}
\end{align*}
\textbf{\eqref{eq-dd-power}:}
For the successor case $\alpha=\beta+1$, we go through
\begin{align*}
\alpha'\dd\beta+1
\text{ implies }
&\alpha'\dde\beta
\tag{by Eq.~\eqref{eq-dd-inductive}}
\\
\text{ implies }
&\omega^{\alpha'}\dde\omega^\beta=\omega^\beta+0
\tag{by ind.\ hyp.}
\\
\text{ implies }
&\omega^{\alpha'}\dde\omega^\beta + \omega^\beta\cdot(\omega_x-1)
\tag{by Eq.~\eqref{eq-dd-sum} and~\eqref{eq-dd-zero}}
\\
\text{ implies }
&\omega^{\alpha'}\dde\omega^\beta \cdot \omega_x = (\omega^{\beta+1})_x
       \ddby{\eqref{eq-dd-def}}\omega^{\beta+1} \;.
\\
\shortintertext{For the limit case $\alpha=\lambda$, this goes through
}
\alpha'\dd\lambda
\text{ implies }
&\alpha'\dde\lambda_x
\tag{by Eq.~\eqref{eq-dd-inductive}}
\\
\text{ implies }
&\omega^{\alpha'}\dde\omega^{\lambda_x}\equalby{\eqref{eq-fund}}(\omega^\lambda)_x
\ddby{\eqref{eq-dd-def}}\omega^\lambda
\: .
\tag{by ind.\ hyp.}
\end{align*}
\end{proof}

\begin{lemma}\label{lem-dd-xy}
  Let $\lambda$ be a limit ordinal in $\Omega$ and $x<y$ in $\Nat$.
  Then $\lambda_x\dd[y]\lambda_y$, and if furthermore $\omega_x>0$,
  then $\lambda_x\dd\lambda_y$.
\end{lemma}
\begin{proof}
  By induction over $\lambda$. Write $\omega_y=\omega_x+z$ for some
  $z\geq 0$ by monotonicity and $\lambda=\gamma+\omega^{\alpha}$ with
  $0<\alpha$.

  If $\alpha=\beta+1$ is a successor, then
  $\lambda_x=\gamma+\omega^\beta\cdot
  \omega_x\dde[y]\gamma+\omega^\beta\cdot \omega_x+\omega^\beta\cdot z$ by Eq.~\eqref{eq-dd-sum} since $0\dde[y]\omega^\beta\cdot z$.
 We conclude by noting that $\lambda_y = \gamma+\omega^{\beta}\cdot
 (\omega_x+z)$; the same arguments also show $\lambda_x\dd\lambda_y$.

 If $\alpha$ is a limit ordinal, then $\alpha_x\dd[y]\alpha_y$ by
 ind.\ hyp., hence $\lambda_x=\gamma+\omega^{\alpha_x}\dd[y]
 \gamma+\omega^{\alpha_y}=\lambda_y$ by Eqs.~\eqref{eq-dd-power}
 (applicable since $\omega_y\geq y>x\geq 0$) and~\eqref{eq-dd-sum}.
 If $\omega_x>0$, then the same arguments show
 $\lambda_x\dd\lambda_y$.
\end{proof}
\noindent
Now, using Eq.~\eqref{eq-dd-inductive} together with
\autoref{lem-dd-xy}, we see
\begin{corollary}
\label{cor-ddx-ddy}
Let $\alpha,\beta$ in $\Omega$ and $x,y$ in $\Nat$.  If $x\leq y$ then
$\alpha\dd\beta$ implies $\alpha\dd[y]\beta$.
\end{corollary}
\noindent
In other words, $\dd \:\subseteq\: \dd[x+1] \:\subseteq\: \dd[x+2]
\:\subseteq\: \cdots$.
\medskip

One sees $\bigl(\bigcup_{x\in\Nat}\dd\bigr) \:=\: <$ over terms in
$\CNF$ as a result of \autoref{lem-Px-dd} that we now set to prove.
Since $\alpha\dde P_x(\gamma)$ directly entails all the other
statements of \autoref{lem-Px-dd}, it is enough to prove:
\begin{claim}
  Let $\alpha,\gamma$ in $\CNF$ and $x$ in $\Nat$.  If $\alpha$ is
  $x$-lean, then
\begin{equation*}
\alpha<\gamma \text{ implies }\alpha\dde
P_x(\gamma)\:.
\end{equation*}
\end{claim}
\begin{proof}
  If $\alpha=0$, we are done so we assume $\alpha>0$ and hence $x>0$,
  thus $\alpha=\sum_{i=1}^m\omega^{\beta_i}\cdot c_i$ with $m>0$ and
  $\omega_x\geq x>0$.  Working with terms in CNF allows us to employ
  the syntactic characterization of $<$ given in
  \eqref{eq-def-ordering}.

  We prove the claim by induction on $\gamma$,
  considering two cases:
  \begin{enumerate}
  \item if $\gamma = \gamma'+1$ is a successor then $\alpha<\gamma$
    implies $\alpha\leq\gamma'$, hence
    $\alpha\ddeby{ih}\gamma'\equalby{\eqref{eq-pred}}P_x(\gamma)$.
  \item if $\gamma$ is a limit, we claim that $\alpha<\gamma_x$, from
    which we deduce
    $\alpha\ddeby{ih}P_x(\gamma_x)\equalby{\eqref{eq-pred}}P_x(\gamma)$.
    We consider three subcases for the claim:
    \begin{enumerate}
    \item if $\gamma=\omega^\lambda$ with $\lambda$ a limit, then
      $\alpha<\gamma$ implies $\beta_1<\lambda$, hence
      $\beta_1\ddeby{ih}P_x(\lambda_x)<\lambda_x$ (since $\beta_1$ is
      $x$-lean). Thus
      $\alpha<\omega^{\lambda_x}=(\omega^\lambda)_x=\gamma_x$.
    \item if $\gamma=\omega^{\beta+1}$ then $\alpha<\gamma$ implies
      $\beta_1<\beta+1$, hence $\beta_1\leq\beta$. Now $c_1\leq x$
      since $\alpha$ is $x$-lean, hence $\alpha <
      \omega^{\beta_1}\cdot(c_1+1) \leq \omega^{\beta_1}\cdot (x+1)
      \leq \omega^{\beta}\cdot (x+1) = (\omega^{\beta+1})_x=\gamma_x$.
    \item if $\gamma=\gamma'+\omega^\beta$ with $0<\gamma',\beta$,
      then either $\alpha\leq\gamma'$, hence $\alpha <
      \gamma'+(\omega^\beta)_x = \gamma_x$, or $\alpha > \gamma'$, and
      then $\alpha$ can be written as $\alpha=\gamma'+\alpha'$ with
      $\alpha'<\omega^\beta$. In that case $\alpha' \ddeby{ih}
      P_x(\omega^\beta) \equalby{\eqref{eq-pred}}
      P_x((\omega^\beta)_x) < (\omega^\beta)_x$, hence $\alpha =
      \gamma'+\alpha' \ltby{\eqref{eq-def-ordering}}
      \gamma'+(\omega^\beta)_x \equalby{\eqref{eq-fund}}
      (\gamma'+\omega^\beta)_x = \gamma_x$.\qedhere
    \end{enumerate}
  \end{enumerate}
\end{proof}

\subsection{Ordinal Indexed Functions}
\label{ssec-ax-FGH}

Let us recall several classical hierarchies from
\citep{cichon83,cichon98}.  Let us fix a unary \emph{control function}
$h:\Nat\to\Nat$; we will see later in \autoref{sub:cont-func} how
hierarchies with different control functions can be related.

\paragraph{Inner Iteration Hierarchies.}

We define the hierarchy $(h^\alpha)_{\alpha\in\Omega}$ by
\begin{align}\tag{see \ref{eq-hardy-h}}
  h^0(x)	&\eqdef x,&
  h^{\alpha+1}(x)&\eqdef h^\alpha(h(x)),&
  h^\lambda(x)	 &\eqdef h^{\lambda_{x}}(x).
\intertext{An example of inner iteration hierarchy is the \emph{Hardy
    hierarchy} $(H^\alpha)_{\alpha\in\Omega}$ defined for $h(x)=x+1$:}
  H^0(x)	&\eqdef x,&
  H^{\alpha+1}(x)&\eqdef H^\alpha(x+1),&
  H^\lambda(x)	 &\eqdef H^{\lambda_x}(x).
\end{align}

\paragraph{Inner and Outer Iteration Hierarchies.}

Again for a unary $h$, we can define a variant
$(h_\alpha)_{\alpha\in\Omega}$ of the inner iteration
hierarchies called the \emph{length hierarchy} by \citet{cichon98} and
defined by
\begin{align}
  \tag{see \ref{eq-length-h}}
  h_0(x)	&\eqdef 0,&
  h_{\alpha+1}(x)&\eqdef 1+h_{\alpha}(h(x)),&
  h_{\lambda}(x) &\eqdef h_{\lambda_{x}}(x).
\intertext{As before, for the successor function $h(x)=x+1$, this
  yields}
  H_0(x)	&\eqdef0,&
  H_{\alpha+1}(x)&\eqdef 1+H_\alpha(x+1),&
  H_\lambda(x)	 &\eqdef H_{\lambda_x}(x).
\end{align}
Those hierarchies are the most closely related to the hierarchies of
functions we define for the length of bad sequences.

\paragraph{Fast Growing Hierarchies.}
Last of all, the \emph{fast growing hierarchy}
$(f_{\alpha})_{\alpha\in\Omega}$ is defined through
\begin{align}\tag{see \ref{eq-fast-h}}
  f_0(x)	&\eqdef h(x),&
  f_{\alpha+1}(x)&\eqdef f_{\alpha}^{\omega_x}(x),&
  f_{\lambda}	 &\eqdef f_{\lambda_{x}}(x),
\intertext{while its standard version (for $h(x)=x+1$) is defined by}
  F_0(x)	&\eqdef x+1,&
  F_{\alpha+1}(x)&\eqdef F_\alpha^{\omega_x}(x),&
  F_\lambda(x)	 &\eqdef F_{\lambda_x}(x).
\end{align}

Lemma~\ref{lem-hierarchies} and a few other properties of these
hierarchies can be proved by rather simple induction arguments.
\begin{lemma}[\autoref{lem-hierarchies}.\ref{lem-hierarchies-1}]
  For all $\alpha>0$ in $\Omega$ and $x$,
  \begin{equation*}
    h_{\alpha}(x)=1+h_{P_{x}(\alpha)}(h(x))\;.
  \end{equation*}
\end{lemma}
\begin{proof}
  By transfinite induction over $\alpha$.  For a successor ordinal
  $\alpha+1$,
  $h_{\alpha+1}(x)=1+h_{\alpha}(h(x))=1+h_{P_x(\alpha+1)}(h(x))$.  For
  a limit ordinal $\lambda$, $h_{\lambda}(x)=h_{\lambda_x}(x)$ is equal
  to $1+h_{P_x(\lambda_x)}(h(x))$ by ind.\ hyp.\ since
  $\lambda_x<\lambda$, which is the same as $1+h_{P_x(\lambda)}(h(x))$
  by definition.
\end{proof}
The same argument shows that for all $\alpha>0$ in $\Omega$ and $x$,
\begin{align}
  \label{eq-pred-hardy}
  h^\alpha(x)&=h^{P_x(\alpha)}(h(x))=h^{P_x(\alpha)+1}(x)\;,\\
  \label{eq-pred-fast}
  f_\alpha(x)&=f^{\omega_x}_{P_x(\alpha)}(x)=f_{P_x(\alpha)+1}(x)\;.
\end{align}

\begin{lemma}[\autoref{lem-hierarchies}.\ref{lem-hierarchies-2}]
  Let $h(x)>x$.	 Then for all $\alpha$ in $\Omega$ and $x$,
  \begin{equation*}
    h_{\alpha}(x)\leq h^\alpha(x)-x\;.
  \end{equation*}
\end{lemma}
\begin{proof}
  By induction over $\alpha$.  For $\alpha=0$,
  $h_0(x)=0=x-x=h^0(x)-x$.  For $\alpha>0$,
  \begin{align*}
    h_{\alpha}(x)
    &=1+h_{P_x(\alpha)}(h(x))\tag{by \autoref{lem-hierarchies}.\ref{lem-hierarchies-1}}\\
    &\leq 1+h^{P_x(\alpha)}(h(x))-h(x)\tag{by ind.\ hyp.\ since
      $P_x(\alpha)<\alpha$}\\
    &\leq h^{P_x(\alpha)}(h(x))-x\tag{since $h(x)>x$}\\
    &=h^\alpha(x)-x\;.\tag{by \eqref{eq-pred-hardy}}
\end{align*}
\end{proof}
Using the same argument, one can check that in particular for
$h(x)=x+1$,
\begin{equation}\label{eq-length-hardy}
  H_\alpha(x)=H^\alpha(x)-x\;.
\end{equation}

\begin{lemma}[\citep{cichon83}]\label{lem-hardy-sum}
  For all $\alpha,\gamma$ in $\Omega$, and $x$,
  \begin{equation*}
    h^{\gamma+\alpha}(x)=h^\gamma\!\left(h^\alpha(x)\right).
  \end{equation*}
\end{lemma}
\begin{proof}
  By transfinite induction on $\alpha$.	 For $\alpha=0$,
  $h^{\gamma+0}(x)=h^\gamma(x)=h^\gamma(h^0(x))$.  For a successor
  ordinal $\alpha+1$,
  $h^{\gamma+\alpha+1}(x)=h^{\gamma+\alpha}(h(x))\equalby{ih}h^\gamma\!\left(h^\alpha(h(x))\right)=h^\gamma\!\left(h^{\alpha+1}(x)\right)$.
  For a limit ordinal $\lambda$,
  $h^{\gamma+\lambda}(x)=h^{(\gamma+\lambda)_x}(x)=h^{\gamma+\lambda_x}(x)\equalby{ih}h^\gamma\!\left(h^{\lambda_x}(x)\right)=h^\gamma\!\left(h^\lambda(x)\right)$.
\end{proof}

\begin{lemma}[\autoref{lem-hierarchies}.\ref{lem-hierarchies-3}]
  For all $\beta$ in $\Omega$, $r<\omega$, and $x$,
  \begin{equation*}
    h^{\omega^\beta\cdot r}(x)=f^r_{\beta}(x)\;.
  \end{equation*}
\end{lemma}
\begin{proof}
In view of \autoref{lem-hardy-sum} and $h^0=f^0=\textit{Id}_\Nat$, it
is enough to prove $h^{\omega^\beta}=f_\beta$, i.e., the $r=1$ case.
We proceed by induction over $\beta$.
\begin{description}
\item[{\boldmath For the base case}.] $h^{\omega^0}(x)=h^1(x)\equalby{\eqref{eq-fast-h}}f_0(x)$.
\item[{\boldmath For a successor $\beta+1$}.] $h^{\omega^{\beta+1}}(x) \equalby{\eqref{eq-hardy-h}}
h^{(\omega^{\beta+1})_x}(x) = h^{\omega^\beta\cdot\omega_x}(x)
\equalby{ih} f_\beta^{\omega_x}(x)
\equalby{\eqref{eq-fast-h}}f_{\beta+1}(x)$.
\item[{\boldmath For a limit $\lambda$}.] $h^{\omega^\lambda}(x) \equalby{\eqref{eq-hardy-h}}
h^{\omega^{\lambda_x}}(x) \equalby{ih} f_{\lambda_x}(x)
\equalby{\eqref{eq-fast-h}}f_{\lambda}(x)$.\qedhere
\end{description}
\end{proof}

\subsection{Pointwise Ordering and Monotonicity}
\label{sub-pointwise}
We set to prove in this section the two equations \eqref{eq-mono-hardy}
and~\eqref{eq-dd-hardy} stated in
\citep[\theoremautorefname~2]{cichon98}.

\begin{lemma}[Equations~\eqref{eq-mono-hardy} and~\eqref{eq-dd-hardy}]
  Let $h$ be a monotone function with $h(x)\geq x$.  Then, for all
  $\alpha,\alpha'$ in $\Omega$ and $x,y$ in $\Nat$,
  \begin{align*}
    \tag{see \ref{eq-mono-hardy}}
    x<y&\text{ implies }h_\alpha(x)\leq h_\alpha(y)\;,\\
    \tag{see \ref{eq-dd-hardy}}
    \alpha'\dd[x]\alpha&\text{ implies }h_{\alpha'}(x)\leq
    h_{\alpha}(x)\;.
  \end{align*}
\end{lemma}
\begin{proof}
  Let us first deal with $\alpha'=0$ for \eqref{eq-dd-hardy}.  Then
  $h_0(x)=0\leq h_\alpha(x)$ for all $\alpha$ and $x$.

  Assuming $\alpha'>0$, the proof now proceeds by simultaneous
  transfinite induction over $\alpha$.
  \begin{description}
  \item[{\boldmath For $0$}.]  Then $h_0(x)=0=h_0(y)$ and
    $\alpha'\dd\alpha$ is impossible.
  \item[{\boldmath For a successor $\alpha+1$}.] For~\eqref{eq-mono-hardy},
    $h_{\alpha+1}(x)=1+h_\alpha(h(x))\leqby{ih\eqref{eq-mono-hardy}}
    1+h_\alpha(h(y))=h_{\alpha+1}(y)$ where the
    ind.\ hyp.\ on~\eqref{eq-mono-hardy} can be applied since $h$ is
    monotone.

    For~\eqref{eq-dd-hardy}, we have
    $\alpha'\preccurlyeq_x\alpha\dd\alpha+1$, hence
$h_{\alpha'}(x)\leqby{ih\eqref{eq-dd-hardy}}h_\alpha(x)
\leqby{ih\eqref{eq-mono-hardy}}h_\alpha(h(x))
\ltby{\eqref{eq-length-h}}h_{\alpha+1}(x)$
    where the ind.\ hyp.\ on~\eqref{eq-mono-hardy} can be applied
    since $h(x)\geq x$.
  \item[{\boldmath For a limit $\lambda$}.]  For~\eqref{eq-mono-hardy},
    $h_\lambda(x)=h_{\lambda_x}(x)\leqby{ih\eqref{eq-mono-hardy}}h_{\lambda_x}(y)\leqby{ih\eqref{eq-dd-hardy}}h_{\lambda_y}(y)=h_\lambda(y)$
    where the ind.\ hyp.\ on~\eqref{eq-dd-hardy} can be applied since
    $\lambda_x\dd[y]\lambda_y$ by
    \autoref{lem-dd-xy}.

    For~\eqref{eq-dd-hardy}, we have
    $\alpha'\preccurlyeq_x\lambda_x\dd\lambda$ with
    $h_{\alpha'}(x)\leqby{ih\eqref{eq-dd-hardy}}h_{\lambda_x}(x)=h_\lambda(x)$.
    \qedhere
  \end{description}
\end{proof}
Essentially the same proof can be carried out to prove the same
monotonicity properties for $h^\alpha$ and $f_\alpha$.	As the
monotonicity properties of $f_\alpha$ will be handy in the remainder
of the section, we prove them now:
\begin{lemma}[\citep{fast}]\label{lem-fast-mono}
  Let $h$ be a function with $h(x)\geq x$.  Then, for all
  $\alpha,\alpha'$ in $\Omega$, $x,y$ in $\Nat$ with $\omega_x>0$,
  \begin{gather}
    \label{eq-fast-mn}
    f_\alpha(x)\geq h(x)\geq x\;.\\
    \label{eq-fast-dd}
    \alpha'\dd\alpha\text{ implies }f_{\alpha'}(x)\leq
    f_{\alpha}(x)\;,\\
    \label{eq-fast-xy}
    x<y\text{ and }h\text{ monotone  imply }f_\alpha(x)\leq f_\alpha(y)\;.
  \end{gather}
\end{lemma}
\begin{proof}[Proof of \eqref{eq-fast-mn}]
  By transfinite induction on $\alpha$.	 For the base case,
  $f_0(x)=h(x)\geq x$ by hypothesis.  For the successor case, assuming
  $f_{\alpha}(x)\geq h(x)$, then by induction on $n>0$,
  $f_\alpha^n(x)\geq h(x)$: for $n=1$ it holds since $f_\alpha(x)\geq
  h(x)$, and for $n+1$ since
  $f^{n+1}_\alpha(x)=f_\alpha(f^n_\alpha(x))\geq f_\alpha(x)$ by ind.\
  hyp.\ on $n$.	 Therefore $f_{\alpha+1}(x)=f^{\omega_x}_\alpha(x)\geq x$ since
  $\omega_x>0$.  Finally, for the limit case,
  $f_\lambda(x)=f_{\lambda_x}(x)\geq x$ by ind.\ hyp.
\end{proof}
\begin{proof}[Proof of \eqref{eq-fast-dd}]
  Let us first deal with $\alpha'=0$.  Then $f_0(x)=h(x)\leq
  f_\alpha(x)$ for all $x>0$ and all $\alpha$ by Eq.~\eqref{eq-fast-mn}.

  Assuming $\alpha'>0$, the proof proceeds by transfinite induction
  over $\alpha$.  The case $\alpha=0$ is impossible.  For the
  successor case, $\alpha'\preccurlyeq_x\alpha\dd\alpha+1$ with
    $f_{\alpha+1}(x)=f^{\omega_x-1}_\alpha\!\left(f_\alpha(x)\right)\geqby{\eqref{eq-fast-mn}}f_\alpha(x)\geqby{ih}f_{\alpha'}(x)$.
    For the limit case, we have
    $\alpha'\preccurlyeq_x\lambda_x\dd\lambda$ with $f_{\alpha'}(x)\leqby{ih}f_{\lambda_x}(x)=f_{\lambda}(x)$.
\end{proof}
\begin{proof}[Proof of \eqref{eq-fast-xy}]
  By transfinite induction over $\alpha$.  For the base case,
  $f_0(x)=h(x)\leq h(y)=f_0(y)$ since $h$ is monotone.	For the
  successor case,
  $f_{\alpha+1}(x)=f^{\omega_x}_\alpha(x)\leqby{\eqref{eq-fast-mn}}f^{\omega_y}_\alpha(x)\leqby{ih}f^{\omega_y}_\alpha(y)=f_{\alpha+1}(y)$
  using $\omega_x\leq\omega_y$.
  For the limit case,
  $f_{\lambda}(x)=f_{\lambda_x}(x)\leqby{ih}f_{\lambda_x}(y)\leqby{\eqref{eq-fast-dd}}f_{\lambda_y}(y)=f_\lambda(y)$,
  where \eqref{eq-fast-dd} can be applied thanks to \autoref{lem-dd-xy}.
\end{proof}

\subsection{Relating Different Assignments of Fundamental Sequences}
\label{sub:fund-seq}
The way we employ ordinal-indexed hierarchies is as \emph{standard}
ways of classifying the growth of functions, allowing to derive
meaningful complexity bounds for algorithms relying on wqo's for
termination.  It is therefore quite important to use a standard
assignment of fundamental sequences in order to be able to compare
results from different sources.  The definition provided in
\eqref{eq-fund} is standard, and the choices $\omega_x=x$ and
$\omega_x=x+1$ can be deemed as ``equally standard'' in the
literature.  We employed $\omega_x=x+1$ in the main text, but the
reader might desire to compare this to bounds using $\omega_x=x$.

A bit of extra notation is needed: we want to compare the length
hierarchies $(h_{s,\alpha})_{\alpha\in\Omega}$ for different choices
of $s$.  Recall that $s$ is assumed to be monotone with $s(x)\geq x$,
which is fulfilled by the identity function $\mathit{id}$.
\begin{lemma}\label{lem:fund-seq}
  Let $\alpha$ in $\Omega$.  If $s(h(x))\leq h(s(x))$ for all $x$,
  then $h_{s,\alpha}(x)\leq h_{\mathit{id},\alpha}(s(x))$ for all $x$.
\end{lemma}
\begin{proof}
  By induction on $\alpha$.  For $0$, $h_{s,0}(x)=0=h_{\mathit{id},0}(s(x))$.
  For a successor ordinal $\alpha+1$,
  $h_{s,\alpha+1}(x)=1+h_{s,\alpha}(h(x))\leqby{ih}1+h_{\mathit{id},\alpha}(s(h(x)))\leqby{\eqref{eq-mono-hardy}}1+h_{\mathit{id},\alpha}(h(s(x)))=h_{\mathit{id},\alpha+1}(s(x))$
  since $s(h(x))\leq h(s(x))$.  For a limit ordinal $\lambda$,
  $h_{s,\lambda}(x)=h_{s,\lambda_x}(x)\leqby{ih}h_{\mathit{id},\lambda_x}(s(x))\leqby{\eqref{eq-dd-hardy}}h_{\mathit{id},\lambda_{s(x)}}(s(x))=h_{\mathit{id},\lambda}(s(x))$
  where $s(x)\geq x$ implies $\lambda_x\dd[s(x)]\lambda_{s(x)}$ by
  \autoref{lem-dd-xy} and allows to apply \eqref{eq-dd-hardy}.
\end{proof}
In particular, for a smooth $h$ and $s(x)=x+1$, $h(x)+1\leq h(x+1)$
and we can apply \autoref{lem:fund-seq} together with
\autoref{prop-hardy} to get a uniform bound using the standard
assignment with $\omega_x=x$ instead of $\omega_x=x+1$: for all
$\alpha$ in $\CNF[\omega^{\omega^\omega}]$ and $n>0$,
\begin{equation}\label{eq:prop-hardy-id}
  M_{\alpha,g}(n)\leq h_{\alpha}(kn+1)
\end{equation}
where $k$ is the leanness of $\alpha$ and $h(x)=x\cdot g(x)$.

\subsection{Relating Different Control Functions}
\label{sub:cont-func}
As in \autoref{sub:fund-seq}, if we are to obtain bounds in terms of a
\emph{standard} hierarchy of functions, we ought to provide bounds for
$h(x)=x+1$ as control.
We are now in position to prove a statement of \citet{cichon83}:
\begin{lemma}[\autoref{lem-hierarchies}.\ref{lem-hierarchies-4}]\label{lem-eb}
  For all $\gamma$ and $\alpha$ in $\Omega$, if $h$ is monotone
  eventually bounded by $F_\gamma$, then $f_\alpha$ is eventually
  bounded by $F_{\gamma+\alpha}$.
\end{lemma}
\begin{proof}
  By hypothesis, there exists $x_0$ (which we can assume wlog.\
  verifies $x_0>0$) s.t.\ for all $x\geq x_0$,
  $h(x)\leq F_\gamma(x)$.  We keep this $x_0$ constant and show by
  transfinite induction on $\alpha$ that for all $x\geq x_0$,
  $f_\alpha(x)\leq F_{\gamma+\alpha}(x)$, which proves the lemma.
  Note that $\omega_x\geq x\geq x_0>0$ and thus that we can apply
  \autoref{lem-fast-mono}.
  \begin{description}
  \item[{\boldmath For the base case $0$}] for all $x\geq x_0$,
    $f_0(x)=h(x)\leq F_\gamma(x)$ by hypothesis.
  \item[{\boldmath For a successor ordinal $\alpha+1$}] we first prove
    that for all $n$ and all $x\geq x_0$,
    \begin{align}
      f^n_\alpha(x)&\leq F_{\gamma+\alpha}^n(x)\;.\label{eq-lem-4-ind-n}\\
      \shortintertext{Indeed, by induction on $n$, for all $x\geq
	x_0$,}
      f_\alpha^0(x)&=x=F^0_{\gamma+\alpha}(x)\notag\\
      f^{n+1}_\alpha(x)&=f_\alpha\!\left(f^n_\alpha(x)\right)\notag\\
      &\leq\tag{by \eqref{eq-fast-xy} on $f_\alpha$ and the ind.\ hyp.\ on $n$}f_\alpha\!\left(F^n_{\gamma+\alpha}(x)\right)\\
      &\leq\tag{since by \eqref{eq-fast-mn} $F_{\gamma+\alpha}(x)\geq
	x\geq x_0$ and by ind.\ hyp.\ on
	$\alpha$}F_{\gamma+\alpha}\!\left(F_{\gamma+\alpha}^n(x)\right)\\
      &=F^{n+1}_{\gamma+\alpha}(x)\;.\notag \shortintertext{Therefore}
      f_{\alpha+1}(x)&=f^x_\alpha(x)\notag\\
      &\leq\tag{by \eqref{eq-lem-4-ind-n} for
	$n=x$}F^x_{\gamma+\alpha}(x)\\
      &=F_{\gamma+\alpha+1}(x)\;.\notag
    \end{align}
  \item[{\boldmath For a limit ordinal $\lambda$}] for all $x\geq
    x_0$,
    $f_\lambda(x)=f_{\lambda_x}(x)\leqby{ih}F_{\gamma+\lambda_x}(x)=F_{(\gamma+\lambda)_x}(x)=F_{\gamma+\lambda}(x)$.\qedhere
  \end{description}
\end{proof}

\begin{remark}
  Observe that the statement of \autoref{lem-eb} is one of the few
  instances in this appendix where ordinal term notations matter.
  Indeed, nothing forces $\gamma+\alpha$ to be an ordinal term in CNF.
  Note that, with the exception of \autoref{lem-Px-dd}, all the
  definitions and proofs given in this appendix are compatible with
  arbitrary ordinal terms in $\Omega$, and not just terms in CNF, so
  this is not a formal issue.

  The issue lies in the intuitive understanding the reader might have
  of a term ``$\gamma+\alpha$'', by interpreting $+$ as the direct sum
  in ordinal arithmetic.  \textbf{This would be a mistake:} in a
  situation where two different terms $\alpha$ and $\alpha'$ denote
  the same ordinal $\mathop{ord}(\alpha)=\mathop{ord}(\alpha')$, we do
  not necessarily have $F_{\alpha}(x)=F_{\alpha'}(x)$: for instance,
  $\alpha=\omega^{\omega^0}$ and $\alpha'=\omega^0+\omega^{\omega^0}$
  denote the same ordinal $\omega$, but $F_\alpha(2)=F_2(2)=2^2\cdot
  2=2^3$ and $F_{\alpha'}(2)=F_3(2)=2^{2^2\cdot 2}\cdot 2^2\cdot
  2=2^{11}$.  The reader is therefore kindly warned that the results
  on ordinal-indexed hierarchies in this appendix should be understood
  \emph{syntactically} on ordinal terms, and not semantically on their
  ordinal denotations.
\end{remark}

The natural question at this point is: how do these new fast growing
functions compare to the functions indexed by terms in CNF?  Indeed,
we should check that e.g.\ $F_{\gamma+{\omega^p}}$ with
$\gamma<\omega^\omega$ is multiply-recursive if our results are to be
of any use.  The most interesting case is the one where $\gamma$ is
finite but $\alpha$ infinite (which is used in the proof of
\autoref{theo-main}):
\begin{lemma}\label{lem-prim-rec}
  Let $\alpha\geq\omega$ and $0<\gamma<\omega$ be in $\CNF$, and
  $\omega_x\eqdef x$.  Then, for all $x$, $F_{\gamma+\alpha}(x)\leq
  F_{\alpha}(x+\gamma)$.
\end{lemma}
\begin{proof}
  We first show by induction on $\alpha\geq\omega$ that
  \begin{claim}
    Let $s(x)\eqdef x+\gamma$.  Then for all $x$,
    $F_{\mathit{id},\gamma+\alpha}(x)\leq F_{s,\alpha}(x)$.
  \end{claim}
  \begin{description}
  \item[{\boldmath base case for $\omega$}]
    $F_{\mathit{id},\gamma+\omega}(x)=F_{\mathit{id},\gamma+x}(x)=F_{s,\omega}(x)$,
  \item[{\boldmath successor case $\alpha+1$}] with $\alpha\geq\omega$, an
  induction on $n$ shows that $F^n_{\mathit{id},\gamma+\alpha}(x)\leq
  F^n_{s,\alpha}(x)$ for all $n$ and $x$ using the ind.\ hyp.\ on
  $\alpha$, thus
  $F_{\mathit{id},\gamma+\alpha+1}(x)=F^x_{\mathit{id},\gamma+\alpha}(x)\leqby{\eqref{eq-fast-mn}}
  F^{x+\gamma}_{\mathit{id},\gamma+\alpha}(x)\leq
  F^{x+\gamma}_{s,\alpha}(x)=F_{s,\alpha+1}(x)$,
\item[{\boldmath limit case $\lambda>\omega$}]
  $F_{\mathit{id},\gamma+\lambda}(x)=F_{\mathit{id},\gamma+\lambda_x}(x)\leqby{ih}F_{s,\lambda_x}(x)\leqby{\eqref{eq-fast-dd}}F_{s,\lambda{x+\gamma}}(x)=F_{s,\lambda}(x)$
  where \eqref{eq-fast-dd} can be applied since
  $\lambda_x\dde\lambda_{x+\gamma}$ by \autoref{lem-dd-xy} (applicable
  since $s(x)=x+\gamma>0$).
\end{description}\medskip

  Returning to the main proof, note that $s(x+1)=x+1+\gamma=s(x)+1$,
  allowing to apply \autoref{lem:fund-seq}, thus for all $x$,
  \begin{align*}
    F_{\mathit{id},\gamma+\alpha}(x)
    &\leq F_{s,\alpha}(x)\tag{by the previous claim}\\
    &=H_s^{\omega^\alpha}(x)\tag{by \autoref{lem-hierarchies}.\ref{lem-hierarchies-3}}\\
    &\leq H_{\mathit{id}}^{\omega^\alpha}(s(x))\tag{by
      \autoref{lem:fund-seq} and \eqref{eq-length-hardy}}\\
    &=F_{\mathit{id},\alpha}(s(x))\;.\tag{by \autoref{lem-hierarchies}.\ref{lem-hierarchies-3}}
  \end{align*}
\end{proof}

\subsection{Classes of Subrecursive Functions}
\label{sub:ext-grz}
We finally consider how some natural classes of recursive functions
can be characterized by closure operations on subrecursive
hierarchies.  The best-known of these classes is the \emph{extended
Grzegorczyk hierarchy} $(\FGH{\alpha})_{\alpha\in\CNF}$ defined by
\citet{fast} on top of the fast-growing hierarchy
$(F_{\alpha})_{\alpha\in\CNF}$ for $\omega_x\eqdef x$.

Let us first provide some background on the definition and properties
of $\FGH{\alpha}$.  The class of functions $\FGH{\alpha}$ is the
closure of the constant, addition, projection---including identity---,
and $F_{\alpha}$ functions, under the operations of
\begin{description}
\item[substitution] if $h_0,h_1,\dots,h_n$ belong to the class, then
so does $f$ if
\begin{equation*}
  f(x_1,\dots,x_n)=h_0(h_1(x_1,\dots,x_n),\dots,h_n(x_1,\dots,x_n))
\:,
\end{equation*}
\item[limited recursion] if $h_1$, $h_2$, and $h_3$ belong to the
  class, then so does $f$ if
\begin{align*}
  f(0,x_1,\dots,x_n)&=h_1(x_1,\dots,x_n)
\:,\\
  f(y+1,x_1,\dots,x_n)&=h_2(y,x_1,\dots,x_n,f(y,x_1,\dots,x_n))
\:,\\
  f(y,x_1,\dots,x_n)&\leq h_3(y,x_1,\dots,x_n)\;.
\end{align*}
\end{description}

The hierarchy is strict for $\alpha>0$, i.e.\
$\FGH{\alpha'}\subsetneq\FGH{\alpha}$ if $\alpha'<\alpha$, because in
particular $F_{\alpha'}\notin\FGH{\alpha}$.  For small
finite values of $\alpha$, the hierarchy characterizes some well-known
classes of functions:
\begin{itemize}
  \item $\FGH{0}=\FGH{1}$ contains all the linear functions, like
    $\lambda x.x+3$ or $\lambda x.2x$,
  \item $\FGH{2}$ contains all the elementary functions, like
    $\lambda x.2^{2^{x}}$,
  \item $\FGH{3}$ contains all the tetration functions, like $\lambda
    x.\tetra{2}{x}$, etc.
\end{itemize}
The union $\bigcup_{\alpha<\omega}\FGH{\alpha}$ is the set of
primitive-recursive functions, while $F_{\omega}$ is an
Ackermann-like non primitive-recursive function; we call
\emph{Ackermannian} such functions that lie in %
$\FGH{\omega}\backslash\bigcup_{\alpha<\omega}\FGH{\alpha}$.
Similarly, $\bigcup_{\alpha<\omega^\omega}\FGH{\alpha}$ is the set of
multiply-recursive functions with $F_{\omega^\omega}$ a non
multiply-recursive function.

The following properties (resp.\ \theoremautorefname~2.10
and \theoremautorefname~2.11 in \citep{fast}) are useful: for all
$\alpha$, unary $f$ in $\FGH{\alpha}$, and $x$,
\begin{align}
  \label{eq-fast-p}
  \alpha>0&\text{ implies }\exists
  p,\,f(x)\leq F^p_\alpha(x)\;,\\
  \label{eq-fast-em}
  &\exists p,\,\forall x\geq p,\,f(x)\leq F_{\alpha+1}(x)\;.
\end{align}
Also note that by \eqref{eq-fast-p}, if a unary function $g$ is
bounded by some function $g'$ in $\FGH{\alpha}$ with $\alpha>0$, then
there exists $p$ s.t.\ for all $x$, $g(x)\leq g'(x)\leq F^p_\alpha(x)$.
Similarly, \eqref{eq-fast-em} shows that for all $x\geq p$, $g(x)\leq
g'(x)\leq F_{\alpha+1}(x)$.

Let us conclude this appendix with the following
slight extension of \autoref{lem-hierarchies}.\ref{lem-hierarchies-4}:
\begin{lemma}\label{lem-fgh}
  For all $\gamma>0$ and $\alpha$, if $h$ is monotone and eventually
  bounded by a function in $\FGH{\gamma}$, then
  \begin{enumerate}[(i)]
  \item\label{lem-4-finite} if $\alpha<\omega$, $f_{\alpha}$ is
    bounded by a function in $\FGH{\gamma+\alpha}$, and
  \item\label{lem-4-infinite} if $\gamma<\omega$ and
    $\alpha\geq\omega$, $f_{\alpha}$ is bounded by a function in
    $\FGH{\alpha}$.
  \end{enumerate}
\end{lemma}
\begin{proof}[Proof of \eqref{lem-4-finite}]
  We proceed by induction on $\alpha<\omega$.
  \begin{description}
  \item[{\boldmath For the base case $\alpha=0$}] we have $f_0=h$ bounded by
    a function in $\FGH{\gamma}$ by hypothesis.
  \item[{\boldmath For the successor case $\alpha=k+1$}] by ind.\ hyp.\
    $f_k$ is bounded by a function in $\FGH{\gamma+k}$, thus
    by \eqref{eq-fast-p} there exists $p$ s.t.\
    $f_k(x)\leq F^p_k(x)$.  By induction on $n$, we deduce
    \begin{align}\label{eq-fast-iter}
    f^n_k(x)&\leq F^{pn}_{\gamma+k}(x)\;;
    \shortintertext{indeed}
    f^0_k(x)&=x=F^0_{\gamma+k}(x)\;,\notag\\
    f^{n+1}_k(x)&=f_k(f^n_k(x))\leqby{ih}f_k(F^{pn}_{\gamma+k}(x))\leqby{\eqref{eq-fast-p}}F^p_{\gamma+k}(F^{pn}_{\gamma+k}(x))=F^{p(n+1)}_{\gamma+k}(x)\;.\notag
    \shortintertext{Therefore,}
    f_{k+1}(x)&=f^x_{k}(x)\leqby{\eqref{eq-fast-iter}}F^{px}_{\gamma+k}(x)\leqby{\eqref{eq-fast-xy}}F^{px}_{\gamma+k}(px)=F_{\gamma+k+1}(px)\;,\notag
    \end{align}
    where the latter function $x\mapsto F_{\gamma+k+1}(px)$ is
    defined by substitution from $F_{\gamma+k+1}$ and $p$-fold
    addition, and therefore belongs to $\FGH{\gamma+k+1}$.\qedhere
  \end{description}
\end{proof}
\begin{proof}[Proof of \eqref{lem-4-infinite}]
  By \eqref{eq-fast-em}, there exists $x_0$ s.t.\ for all $x\geq x_0$,
  $h(x)\leq F_{\gamma+1}(x)$.  By
  \autoref{lem-hierarchies}.\ref{lem-hierarchies-4} and
  \autoref{lem-prim-rec}, $f_\alpha(x)\leqby{\eqref{eq-fast-xy}}f_{\alpha}(x+x_0)\leq F_{\alpha}(x+x_0+\gamma+1)$ for
  all $x$, where the latter function $x\mapsto
  F_{\alpha}(x+x_0+\gamma+1)$ is in $\FGH{\alpha}$.
\end{proof}

%% file: sec-comments.tex
\section{Additional Comments}\label{app-comments}
We gather in this appendix several additional remarks comparing some
of the more technical aspects of the main text with the literature.

\subsection{Maximal Order Types}\label{app-ot}

\paragraph{Definitions of Maximal Order Types.}
Our definition of $o(A)$ in \autoref{sec-otypes} is the same as that of
the \emph{maximal order type} of the wpo $A$, which is defined as the
sup of all the order types of the linearizations of
$\tup{A;\leq}$~\citep{dejongh77}, or equivalently as the height of the
tree of bad sequences of $A$~\citep{hasegawa94}---this is not a mere
coincidence, as we will see at the end of the section when introducing
reifications.

Consider a well partial order $\tup{A;\leq}$.  The first definition of
the maximal order type of $A$ is through \emph{linearizations},
i.e. linear orderings $\leq'$ extending $\leq$:
\begin{align}
  o(A)&=\sup\{|A;{\leq'}|\;\mid\;\leq'\text{ is a linearization of
  }\leq\}\;.\tag{\citep[\definitionautorefname~1.4]{dejongh77}}
\intertext{%
This definition uses the fact that well-linear orders and ordinal
terms in $\CNF$ can be identified.
For the second definition, organize the set of bad
sequences over $\tup{A;\leq}$ as a prefix tree $\mathrm{Bad}$, and
associate an ordinal $|\sigma|$ to each node $\sigma$ respecting%
}
|\sigma|&=\sup\{|\sigma'|+1\mid \sigma'\text{ is an immediate successor
of }\sigma\}\;.\notag
\intertext{Write $|\mathrm{Bad}|$ for the root ordinal:%
}
  o(A)&=|\mathrm{Bad}|\;.\tag{\citep[\definitionautorefname~2.7]{hasegawa94}}
\end{align}

\paragraph{Bijection With Algebra.}
The bijection between exponential nwqo's and ordinal terms in $\CNF[\omega^{\omega^\omega}]$ is
not extremely surprising.  A bijection for an algebra on wpo's with
fixed points instead of Kleene star is shown to hold by
\citet{hasegawa94}, and applied to Kruskal's Theorem.  The novelty in
\autoref{sec-otypes} is that everything also works for \emph{normed}
wqo's.

Finally note that using different algebraic operators can easily break
this bijection.  For instance, $\seg{p}$, the $p$-element initial
segment of $\Nat$, has order type $p=o(\Upgamma_p)$, but for $p\geq 2$
the two nwqo's $\seg{p}$ and $\Upgamma_p$ are not isomorphic.

\paragraph{Reifications.}
Equation~\ref{eq-main-inequality} can be viewed as a controlled variant of
the reification techniques usually employed to prove upper
bounds on maximal order types~\citep{simpson88,hasegawa94}.

A \emph{reification} of a partial order $\tup{A;\leq}$ by an ordinal
$\alpha$ is a map $\mathrm{Bad}\to\alpha+1$ s.t.\ if $\sigma'$ is a
suffix of $\sigma$, then
$r(\sigma')<r(\sigma)$~\citep[Def.~4.1]{simpson88}.  If there exists
such a reification, then $o(A)<\alpha+1$ and $\tup{A;\leq}$ is a wpo.

Given a normed partial order $A$ and any bad sequence
$\xxx=x_0,x_1,\ldots$, we can define a control
$g(x)=\max\{|x_{i+1}|+1\mid x=|x_i|+1\}$ (remember that $\len{.}_A$ is
proper) such that $\xxx$ is $(g,|x_0|+1)$-controlled, and use
\eqref{eq-main-inequality} to associate with each (bad) suffix
$x_{i+1},x_{i+2},\ldots$ an ordinal $\alpha_{i+1}$ that maximizes
$L_{C(\alpha_i)}$ in \eqref{eq-main-inequality}.  Since
$\alpha\mathrel{\partial_n}\alpha'$ implies $\alpha>\alpha'$ for all
$n$, this mapping yields a well-founded sequence
$\alpha_0>\alpha_1>\cdots$ of ordinals, of length at most
$\alpha_0=o(A)$.  While not a reification stricto sensu, this
association of decreasing ordinals to each suffix of any bad sequence
$\xxx$ implies every bad sequence $\xxx$ to be finite and $A$ to be a
wqo, i.e., \eqref{eq-main-inequality} implies Higman's Lemma in the
finite alphabet case.  A second consequence is that no choice of
$o(A)$ smaller than the maximal order type of $A$ can be compatible
with an inequality like \eqref{eq-main-inequality}, since the
particular linearization that gave rise to $o(A)$ yields one
particular bad sequence.

\subsection{Comparisons with the Literature}
\label{ax:comps}
We provide some elements of comparison between our bounds
and similar bounds found in the literature.

\paragraph{Lower Bounds.} Let us compare our bound with
the lower bound of \citet{lowhigman}, who constructs a
$(g,n)$-controlled bad sequence $\xxx=x_0,x_1,\ldots$ of length
\begin{equation*}
g_{\omega^{\omega^{p-1}}}(n)\leq \badd{\Upgamma_{p}^\ast,g}(n^p)\;
\end{equation*}
for $\omega_x=x$.

The $n^p$ bound on the length of $x_0$ in this sequence results from
an alternative definition of the norm over $\Upgamma_{p}^\ast$.	 Let
$\Upgamma_{p+1}=\{a_1,\ldots,a_p,a_{p+1}\}$, and
$\pi_p:\Upgamma_{p+1}^\ast\to\Upgamma_p^\ast$ be the projection
defined by $\pi_p(a_{p+1})=\varepsilon$ (the empty string) and
$\pi_p(a_i)=a_i$ for all $1\leq i\leq p$.  The norm
$\|.\|_{\Upgamma_{p+1}}$ is defined by \citet{lowhigman} for all $x$ in
$\Upgamma_{p+1}$ by
\begin{equation*}
\|x\|_{\Upgamma_{p+1}}=\max(\{\|\pi_p(x)\|_{\Upgamma_p}\}\cup\{|y|\mid\exists
z,z'\in\Upgamma_{p+1},x=zyz'\wedge y\in\{a_{p+1}\}^\ast\})
\end{equation*} for $p\geq
0$ and $\|x\|_{\Upgamma_0}=0$.	For instance,
$\|a_2^ia_1a_1a_2^ja_1^k\|_{\Upgamma_2}=\max(i,j,k+2)$.	 Observe that,
if a sequence in
$(\Upgamma_{p}^\ast,|.|_{\Upgamma_{p}};\leq_{\Upgamma_{p}})$ is
$(g,n)$-controlled, then seeing it as a sequence in
$(\Upgamma_{p}^\ast,\|.\|_{\Upgamma_{p}};\leq_{\Upgamma_{p}})$,
it remains $(g,n)$-controlled; thus despite being more involved this
norm could be used seamlessly in applications.	But, most importantly,
using this new norm does not break \eqref{eq-residual-finite-star},
and the entire analysis we conducted still holds.  Thus, by
\eqref{eq:prop-hardy-id}, in
$(\Upgamma_{p}^\ast,\|.\|_{\Upgamma_{p}};\leq_{\Upgamma_{p}})$,
\begin{equation*}
  g_{\omega^{\omega^{p-1}}}(n)\leq \badd{\Upgamma_{p}^\ast,g}(n)\leq
  h_{\omega^{\omega^{p-1}}}((p-1)n+1)\;.
\end{equation*}

\paragraph{Upper Bounds.}
Because previous authors employed various modified subrecursive
hierarchies (as we do with $h_\alpha$) but did not provide any
translation into the standard ones (like our \autoref{theo-main}),
comparing our bound with theirs is very difficult.  \Citet{cichon98}
show a $g_{\alpha_p}\!(n)$ upper bound where $\alpha_p$ is an ordinal
with a rather complex definition
\citep[see][\sectionautorefname~8]{cichon98}.
\Citet[\corollaryautorefname~6.3]{weiermann94} assumes $g(x+1)=g(x)+d$
for some constant $d$ and shows a
$\bar{H}^{\omega^{\omega^{p-1}}}\!\!\!\left((4+p+12\cdot(n+2+d))^3\right)$
upper bound using a modified Hardy hierarchy
$(\bar{H}^\alpha)_\alpha$; it is actually not clear whether this would
be eventually bounded by $F_{\omega^{p-1}}$.  See the next section for
a discussion of how the techniques of \citet{buchholtz94,weiermann94}
could be applied to our case.

\subsection{Normed Systems of Fundamental Sequences}
\label{ax-nhardy}
We discuss in this subsection an alternative proof of
\autoref{prop-hardy} (with an additional hygienic condition on $g$),
which relies on the work of \citet{buchholtz94} on alternative
definitions of hierarchies, and in particular on their
\theoremautorefname~4.  The proof reuses some results given in
Appendix~\ref{ax-hardy} (namely \eqref{eq-mono-hardy} and
\autoref{coro-partial-lean}), and the interplay between leanness and
the predecessor function (\autoref{lem-Px-dd}).  Throughout the
section, fix $\omega_x\eqdef x+1$.

Let us first define a \emph{norm} $N$ over $\CNF$ by
\begin{equation}
N\alpha=\min\{k\in\Nat\mid \alpha\text{ is $k$-lean}\}\;.
\end{equation}
One can verify
\begin{align}\label{eq-norm-bach}
  \forall\alpha,\,N0\leq N\alpha&\text{ and
  }\forall\alpha,\,N(\alpha+1)\leq N(\alpha)+1\;.
\end{align}
Note that for each $k$ and each $\alpha$ in $\CNF$ (thus
$\alpha<\ezero$), there are only finitely many ordinal terms
$\alpha'<\alpha$ with $N\alpha'\leq k$.  This is useful in definitions
like Eq.~\eqref{eq-nhardy} below, where it ensures that the $\max$
operation is applied to a finite set.

Given a monotone control function $g$, define an alternative length
hierarchy by
\begin{align}\label{eq-nhardy}
  \tilde{H}_0(x)&=0,\,&
  \tilde{H}_\alpha(x)&=\max\{1+\tilde{H}_{\alpha'}(g(x))\mid \alpha'<\alpha\wedge
    N\alpha'\leq (N\alpha)\cdot x\}\,.
\end{align}
One easily proves, by induction over $\alpha$, that each
$\tilde{H}_\alpha$ is monotone.  By \autoref{coro-partial-lean}, if
$\alpha'\mathrel{\partial_n}\alpha$, then $N\alpha'\leq(N\alpha)n$,
and as seen in Appendix~\ref{app-deriv-is-wf} $\alpha'<\alpha$, thus
for all $\alpha$ and all $n$,
\begin{equation}\label{eq-ubound-nhardy}
  M_\alpha(n)\leq \tilde{H}_\alpha(n)\;.
\end{equation}
We could stop here: after all, which hierarchy definition constitutes
the appropriate one is debatable.  Nevertheless, we shall continue
toward a more ``standard'' understanding (see \eqref{eq-nhardy-pred})
of the alternative length hierarchy $(\tilde{H}_\alpha)_\alpha$
defined in \eqref{eq-nhardy}, from which \autoref{prop-hardy} will be
quite easy to derive (see \eqref{eq-nhardy-ind}).

\paragraph{An alternative assignment of fundamental sequences.}
Consider the following alternative assignment of fundamental sequences
(also defined on zero and successor ordinals):
\begin{align}\label{eq-bachman}
  [0]_x&=0,\,&
  [\alpha]_x&=\max\{\alpha'<\alpha\mid
  N\alpha'\leq(N\alpha)\cdot x\}\;.
\end{align}
This almost fits the statement of
\citep[\theoremautorefname~4]{buchholtz94}, which defines fundamental
sequences using a ``$N\alpha'\leq p(\alpha+x)$'' condition for a
suitable function $p$, instead of ``$N\alpha'\leq(N\alpha)\cdot x$''
as in \eqref{eq-bachman}.  Nevertheless, we can follow the proof of
\citep[\theoremautorefname~4]{buchholtz94} and adapt it quite easily
to our case.%

\begin{lemma}\label{lem-bachman}
  Let $g(x)\geq 2x$ for all $x$ be a control function.  Then, for all
  $\alpha>0$ in $\CNF$ and all $x$,
  \begin{equation*}\label{eq-maxbachman}
    \tilde{H}_\alpha=1+\tilde{H}_{[\alpha]_x}(g(x))\;.
  \end{equation*}
\end{lemma}
\begin{proof}
One inequality is immediate since $[\alpha]_x$ verifies the conditions of \eqref{eq-nhardy} by
definition, thus
\begin{equation}
  1+\tilde{H}_{[\alpha]_x}(g(x))\leq\max\{1+\tilde{H}_{\alpha'}(g(x))\mid
\alpha'<\alpha\wedge N\alpha'\leq(N\alpha)\cdot x\}\;.
\end{equation}
The proof of the converse inequality is more involved.  Let us first
show the following:
\begin{claim}
  Let $g(x)\geq 2x$ for all $x$ be a control function.  If
  $\alpha'<[\alpha]_x$ and $N\alpha'\leq(N\alpha)\cdot x$, then
  \begin{equation}\label{eq-nbach-ind}
    N\alpha'\leq N[\alpha]_x\cdot g(x)\;.
  \end{equation}
\end{claim}
\begin{proof}[Proof of \eqref{eq-nbach-ind}]
First note that, for a successor ordinal,
\begin{equation}
  [\alpha+1]_x=\alpha
\end{equation}
since $\alpha$ satisfies the conditions of \eqref{eq-bachman} and is
the maximal ordinal to do so.  Thus
\begin{equation}\label{eq-nbach-succ}
  N[\alpha+1]_x=N\alpha\;.
\end{equation}
Also note that, if $\lambda$ is a limit, then
\begin{equation}\label{eq-nbach-lim}
  N[\lambda]_x=(N\lambda)\cdot x\;.
\end{equation}
Indeed, assume $N[\lambda]_x\neq(N\lambda)\cdot x$.  By
definition \eqref{eq-bachman}, we have $[\lambda]_x<\lambda$ and
$N[\lambda]_x\leq(N\lambda)\cdot x$, so this would mean
$N[\lambda]_x<(N\lambda)\cdot x$.  But in that case
$[\lambda]_x+1<\lambda$ since $[\lambda]_x<\lambda$ and $\lambda$ is a
limit, and $N([\lambda]_x+1)\leq N[\lambda]_x+1\leq(N\lambda)\cdot x$
by \eqref{eq-norm-bach}, hence $[\lambda]_x+1$ also satisfies the
conditions of \eqref{eq-bachman} with $[\lambda]_x<[\lambda]_x+1$, a
contradiction.

Let us now prove the claim itself.  Note that $\alpha'<[\alpha]_x$
implies $[\alpha]_x>0$.  If $\alpha$ is a limit ordinal $\alpha''+1$,
then
\begin{align*}
  N\alpha'&<(N(\alpha''+1))\cdot x\tag{by hyp.}\\
  &\leq (N\alpha'' + 1)\cdot x\tag{by \eqref{eq-norm-bach}}\\
  &= (N[\alpha]_x+1)\cdot x\tag{by \eqref{eq-nbach-succ}}\\
  &\leq (N[\alpha]_x)\cdot 2x\tag{since $[\alpha]_x>0$}\\
  &\leq (N[\alpha]_x)\cdot g(x)\;.\tag{since $g(x)\geq 2x$}
\end{align*}
If $\alpha$ is a limit ordinal, then
\begin{align*}
  N\alpha'&<(N\alpha)\cdot x\tag{by hyp.}\\
  &=(N[\alpha]_x)\cdot x\tag{by \eqref{eq-nbach-lim}}\\
  &\leq (N[\alpha]_x)\cdot g(x)\;.\tag{since $g(x)\geq x$}
\end{align*}
\end{proof}

Returning to the proof of \autoref{lem-bachman}, we show by induction
on $\alpha$ that
\begin{claim}
\begin{equation}\label{eq-maxbachman-ind}
  \alpha'<\alpha\text{ and }N\alpha'\leq (N\alpha)\cdot x\text{ imply
  }1+\tilde{H}_{\alpha'}(g(x))\leq\tilde{H}_\alpha(x)\;.
\end{equation}
\end{claim}
\begin{proof}[Proof of \eqref{eq-maxbachman-ind}]
We have $\tilde{H}_\alpha(x)=1+\tilde{H}_{[\alpha]_x}(g(x))$ by
definition, and by the hypotheses of \eqref{eq-maxbachman-ind} we get
$\alpha'\leq[\alpha]_x$.  If $\alpha'=[\alpha]_x$, then
\eqref{eq-maxbachman-ind} holds.  Otherwise, i.e.\ if
$\alpha'<[\alpha]_x$, \eqref{eq-nbach-ind} shows
$N\alpha'\leq (N[\alpha]_x)\cdot g(x)$, and we can apply the
ind.\ hyp.\ on $[\alpha]_x$:
\begin{align*}
  1+\tilde{H}_{\alpha'}(g(x))&<2+\tilde{H}_{\alpha'}(g^2(x))\tag{by
    monotonicity of $g$ and $\tilde{H}_{\alpha'}$}\\
  &\leq 1+\tilde{H}_{[\alpha]_x}(g(x))\tag{by ind.\ hyp.}\\
  &=\tilde{H}_\alpha(x)\;.\qedhere
\end{align*}
\end{proof}
The previous claim implies the desired inequality and concludes the
proof of \autoref{lem-bachman}.
\end{proof}

\paragraph{Relating with Predecessors.}
We first revisit the relationship between leanness and predecessor
computations (this also provides an alternative proof of
\autoref{lem-Px-dd}).
\begin{lemma}\label{lem-pred-lean}
  Let $\alpha$ be in $\CNF$ and $k$ in $\Nat$.  If $\alpha$ is
  $k$-lean, then $P_k(\alpha)$ is also $k$-lean, and furthermore
  $P_k(\alpha)=\max\{\alpha'\text{ $k$-lean}\mid
  \alpha'\dd[k]\alpha\}$.
\end{lemma}
\begin{proof}
  Let us introduce a slight variant of $k$-lean ordinals: let
  $\alpha=\omega^{\beta_1}\cdot c_1+\cdots+\omega^{\beta_m}\cdot c_m$
  be an ordinal in $\CNF$ with $\alpha>\beta_1>\cdots>\beta_m$ and
  $\omega>c_1,\ldots,c_m>0$.  We say that $\alpha$ is \emph{almost
    $k$-lean} if either (i)~$c_m=k+1$ and both
  $\sum_{i<m}\omega^{\beta_i}$ and $\beta_m$ are $k$-lean, or
  (ii)~$c_m\leq k$, $\sum_{i<m}\omega^{\beta_i}$ is $k$-lean, and
  $\beta_m$ is almost $k$-lean.  Note that an almost $k$-lean ordinal
  term is \emph{not} $k$-lean.  Here are several properties of note on
  almost $k$-lean ordinals:
  \begin{claim}\label{cl-almost-lean}
    If $\lambda$ is $k$-lean, then $\lambda_k$ is almost $k$-lean.
  \end{claim}
  By induction on $\lambda$, letting $\lambda=\omega^{\beta_1}\cdot
  c_1+\cdots+\omega^{\beta_m}\cdot c_m$ as above.  If $\beta_m$ is a
  successor ordinal $\beta+1$ (thus $\beta$ is $k$-lean),
  $\lambda_k=\omega^{\beta_1}\cdot c_1+\cdots+\omega^{\beta_m}\cdot
  (c_m-1)+\omega^{\beta}\cdot (k+1)$ is almost $k$-lean.  If $\beta_m$ is
  a limit ordinal, $\lambda_k=\omega^{\beta_1}\cdot
  c_1+\cdots+\omega^{\beta_m}\cdot (c_m-1)+\omega^{(\beta_m)_k}$ is
  $k$-lean by ind.\ hyp.\ on $\beta_m$.
  
  \begin{claim}\label{cl-almost-succ}
    If $\alpha+1$ is almost $k$-lean, then $\alpha$ is $k$-lean.
  \end{claim}  
  If $\alpha+1=\omega^{\beta_1}\cdot
  c_1+\cdots+\omega^{\beta_m}\cdot c_m$ as above, it means
  $\beta_m=0$, thus we are in case~(i) of almost $k$-lean ordinals
  with $c_m=k+1$, and $\alpha=\omega^{\beta_1}\cdot
  c_1+\cdots+\omega^{\beta_m}\cdot(c_m-1)$ is $k$-lean.
  
  \begin{claim}\label{eq-almost-limit}
    If $\lambda$ is almost $k$-lean, then $\lambda_k$ is almost
    $k$-lean.
  \end{claim}
  By induction on $\lambda$, letting $\lambda=\omega^{\beta_1}\cdot
  c_1+\cdots+\omega^{\beta_m}\cdot c_m$ as above.
  \begin{description}
  \item[{\boldmath If $\beta_m$ is a successor ordinal $\beta+1$,}]
  $\lambda_k=\omega^{\beta_1}\cdot c_1+\cdots+\omega^{\beta_m}\cdot
  (c_m-1)+\omega^{\beta}\cdot (k+1)$, and either (i)~$c_m=k+1$ and $\beta_m$
  is $k$-lean, and then $\lambda_k$ also verifies~(i), or
  (ii)~$c_m\leq k$
  and $\beta+1$ is almost $k$-lean and thus $\beta$ is $k$-lean by the
  previous claim, and $\lambda_k$ is again almost $k$-lean verifying
  condition~(i).
  \item[{\boldmath If $\beta_m$ is a limit ordinal,}] then
  $\lambda_k=\omega^{\beta_1}\cdot c_1+\cdots+\omega^{\beta_m}\cdot
  (c_m-1)+\omega^{(\beta_m)_k}$.  Either (i)~$c_m=k+1$ and $\beta_m$ is
  $k$-lean, and by the previous claims $(\beta_m)_k$ is almost
  $k$-lean and $\lambda_k$ is almost $k$-lean by condition~(ii), or
  (ii)~$c_m\leq k$ and $\beta_m$ is almost $k$-lean, and by ind.\ hyp.\
  $(\beta_m)_k$ is almost $k$-lean, and $\lambda_k$ almost $k$-lean by
  condition~(ii).
  \end{description}
  The proof of the lemma is then straightforward by applications of
  the previous claims and the definition of the predecessor function
  in \eqref{eq-pred}.
\end{proof}

\begin{lemma}\label{lem-predbachman}
  For all $\alpha>0$ in $\CNF$ and $x$,
\begin{equation*}
  [\alpha]_x=P_{N\alpha\cdot x}(\alpha)\;.
\end{equation*}
\end{lemma}
\begin{proof}
  First observe that $P_{N\alpha\cdot x}(\alpha)<\alpha$, and
  furthermore $N(P_{N\alpha\cdot x}(\alpha))=N\alpha$ by
  \autoref{lem-pred-lean}, hence $P_{N\alpha\cdot x}(\alpha)$
  satisfies the conditions of \eqref{eq-bachman}: $[\alpha]_x\geq
  P_{N\alpha\cdot x}(\alpha)$.

  Conversely, let $\alpha'$ be such that $\alpha'<\alpha$ and
  $N\alpha'\leq N\alpha\cdot x$, i.e.\ $\alpha'$ is $(N\alpha\cdot
  x)$-lean.  By \autoref{lem-Px-dd}, $\alpha'\dd[N\alpha\cdot
  x]\alpha$.  Still by \autoref{lem-Px-dd}, $\alpha'\leq
  P_{N\alpha\cdot x}(\alpha)$, hence $[\alpha]_x\leq P_{N\alpha\cdot
    x}(\alpha)$.
\end{proof}

\paragraph{Wrapping up.}
Combining \autoref{lem-bachman} and \autoref{lem-predbachman}, we obtain
that for $g$ monotone with $g(x)\geq 2x$ and for all $\alpha>0$ and all
$x$,
\begin{equation}\label{eq-nhardy-pred}
  \tilde{H}_\alpha(x)=1+\tilde{H}_{P_{N\alpha\cdot x}(\alpha)}(g(x))\;.
\end{equation}

Let us show by induction on $\alpha$ that
\begin{equation}\label{eq-nhardy-ind}
  \tilde{H}_\alpha(x)\leq h_\alpha(N\alpha\cdot x)
\end{equation}
where $h(x)=x\cdot g(x)$.  Proposition~\ref{prop-hardy} will then
follow from \eqref{eq-ubound-nhardy} and \eqref{eq-nhardy-ind}.  For
$\alpha=0$, $\tilde{H}_0(x)=0=h_0(0\cdot x)$.  For the induction
step with $\alpha>0$,
\begin{align*}
  \tilde{H}_\alpha(x)
  &=1+\tilde{H}_{P_{N\alpha\cdot x}(\alpha)}(g(x))\tag{by
    \eqref{eq-nhardy-pred}}\\
  &\leq
  1+h_{P_{N\alpha\cdot x}(\alpha)}\!\left(N(P_{N\alpha\cdot
      x}(\alpha))\cdot g(x)\right)\tag{by ind.\ hyp.\ since
    $P_{N\alpha\cdot x}(\alpha)<\alpha$}\\
  &=1+h_{P_{N\alpha\cdot x}(\alpha)}\!\left(N\alpha\cdot x\cdot g(x)\right)\tag{by
    \autoref{lem-pred-lean}}\\
  &\leq
  1+h_{P_{N\alpha\cdot x}(\alpha)}\!\left(N\alpha\cdot x\cdot
    g(N\alpha\cdot x)\right)\tag{since $N\alpha>0$ by monotonicity of
    $g$ and $h_{P_{N\alpha\cdot x}(\alpha)}$}\\
  &=h_\alpha(N\alpha\cdot x)\;.
\end{align*}

\paragraph{Comparisons.}
\Citet{weiermann94} expresses his upper bound in terms of an
alternative Hardy hierarchy $(\bar{H}^\alpha)_\alpha$ defined by
\begin{align}
  \bar{H}^0(x)&=0,\,&
  \bar{H}^\alpha(x)&=\max\{\bar{H}^{\alpha'}(x+1)\mid
  \alpha'<\alpha\wedge N'\alpha'\leq 2^{N'\alpha+x}\}\;,
\end{align}
where the norm function $N'$ measures the ``depth and width'' of
ordinal terms and is defined by
\begin{align}
  N'0&=0,\,&N'(\omega^{\beta_1}+\cdots+\omega^{\beta_m})&=1+\max\{m,N'\beta_1,\ldots,N'\beta_m\}\;.
\end{align}
Defining an assignment of fundamental sequences by
\begin{align}
  \{0\}_x&=0&
  \{\alpha\}_x&=\max\{\alpha'<\alpha\mid N'\alpha'\leq
  2^{N'\alpha+x}\}
\end{align}
we obtain by \citep[\theoremautorefname~4]{buchholtz94} that for
$\alpha>0$
\begin{equation}
  \bar{H}^\alpha(x)=\bar{H}^{\{\alpha\}_x}(x+1)\;.
\end{equation}
However the similarity with the previous developments stops here: with
this Hardy hierarchy and this norm, there is no direct relationship
with predecessors as in \autoref{lem-predbachman}: Consider for
instance $\alpha=\omega\cdot n$ for some $n>1$, and thus with
$N'\alpha=n$, then
$P_{2^{N'\alpha+x}+1}(\alpha)=\omega\cdot(n-1)+2^{n+x}$ with norm
$N'(\omega\cdot(n-1)+2^{n+x})=2^{n+x}+(n-1)>2^{n+x}$, thus
$P_{2^{N'\alpha+x}+1}(\alpha)\neq\{\alpha\}_x$.  This makes the
translation of bounds in terms of $\bar{H}^\alpha$ into more
``standard'' hierarchies significantly harder (in fact
\citeauthor{weiermann94} does not provide any)---our particular
variations of the ideas present in the work of \citet{buchholtz94}
might actually be of independent interest.

%% file: appendices_article.bbl
\providecommand{\Cerans}{\v{C}er\=ans}

%% file: article.bbl
\providecommand{\Cerans}{\v{C}er\=ans}
\begin{thebibliography}{25}
\providecommand{\natexlab}[1]{#1}
\providecommand{\href}[2]{#2}
\providecommand{\nolinkurl}[1]{#1}
\providecommand{\natconfdetails}[2][]{}
\def\UrlBreaks{\do\@\do\\\do\/\do\!\do\|\do\;\do\>\do\]%
 \do\,\do\?\do\'\do+\do\=\do\#}
\providecommand{\urlstyle}[1]{}
\urlstyle{same}

\bibitem[\begingroup Abdulla\endgroup\ \textnormal{et~al.}(2000)Abdulla,
  \Cerans, Jonsson, and Tsay]{wqo}
Abdulla, P.A., \Cerans, K., Jonsson, B., and Tsay, Y.K., 2000.
\newblock Algorithmic analysis of programs with well quasi-ordered domains.
\newblock \emph{Information and Computation}, 160\penalty0 (1--2):\penalty0
  109--127.
\newblock \href{http://dx.doi.org/10.1006/inco.1999.2843}
  {\nolinkurl{doi:10.1006/inco.1999.2843}}.

\bibitem[\begingroup Atig\endgroup\ \textnormal{et~al.}(2010)Atig, Bouajjani,
  Burckhardt, and Musuvathi]{wmreach}
Atig, M.F., Bouajjani, A., Burckhardt, S., and Musuvathi, M., 2010.
\newblock On the verification problem for weak memory models.
\newblock In \emph{POPL 2010}\natconfdetails[\emph{37th}]{\emph{Annual
  Symposium on Principles of Programming Languages}}, pages 7--18. ACM Press.
\newblock \href{http://dx.doi.org/10.1145/1706299.1706303}
  {\nolinkurl{doi:10.1145/1706299.1706303}}.

\bibitem[\begingroup Buchholz\endgroup\ \textnormal{et~al.}(1994)Buchholz,
  Cicho\'n, and Weiermann]{buchholtz94}
Buchholz, W., Cicho\'n, E.A., and Weiermann, A., 1994.
\newblock A uniform approach to fundamental sequences and hierarchies.
\newblock \emph{Mathematical Logic Quaterly}, 40\penalty0 (2):\penalty0
  273--286.
\newblock \href{http://dx.doi.org/10.1002/malq.19940400212}
  {\nolinkurl{doi:10.1002/malq.19940400212}}.

\bibitem[\begingroup Chambart\endgroup\ and\ \begingroup
  Schnoebelen\endgroup(2007)Chambart and Schnoebelen]{CS-fsttcs07}
Chambart, P. and Schnoebelen, {\relax Ph}., 2007.
\newblock {Post} embedding problem is not primitive recursive, with
  applications to channel systems.
\newblock In Arvind, V. and Prasad, S., editors, \emph{FSTTCS
  2007}\natconfdetails[\emph{27th}]{\emph{IARCS Annual Conference on
  Foundations of Software Technology and Theoretical Computer Science}}, volume
  4855 of \emph{Lecture Notes in Computer Science}, pages 265--276. Springer.
\newblock \href{http://dx.doi.org/10.1007/978-3-540-77050-3_22}
  {\nolinkurl{doi:10.1007/978-3-540-77050-3_22}}.

\bibitem[\begingroup Chambart\endgroup\ and\ \begingroup
  Schnoebelen\endgroup(2008)Chambart and Schnoebelen]{CS-lics08}
Chambart, P. and Schnoebelen, {\relax Ph}., 2008.
\newblock The ordinal recursive complexity of lossy channel systems.
\newblock In \emph{LICS 2008}\natconfdetails[\emph{23rd}]{\emph{Annual IEEE
  Symposium on Logic in Computer Science}}, pages 205--216. IEEE.
\newblock \href{http://dx.doi.org/10.1109/LICS.2008.47}
  {\nolinkurl{doi:10.1109/LICS.2008.47}}.

\bibitem[\begingroup Chambart\endgroup\ and\ \begingroup
  Schnoebelen\endgroup(2010)Chambart and Schnoebelen]{CS-countpep}
Chambart, P. and Schnoebelen, {\relax Ph}., 2010.
\newblock Pumping and counting on the {Regular {Post} Embedding Problem}.
\newblock In Abramsky, S., Gavoille, C., Kirchner, C., auf~der Heide, F.M., and
  Spirakis, P.G., editors, \emph{ICALP
  2010}\natconfdetails[\emph{37th}]{\emph{International Colloquium on Automata,
  Languages and Programming}}, volume 6199 of \emph{Lecture Notes in Computer
  Science}, pages 64--75. Springer.
\newblock \href{http://dx.doi.org/10.1007/978-3-642-14162-1_6}
  {\nolinkurl{doi:10.1007/978-3-642-14162-1_6}}.

\bibitem[\begingroup Cicho\'n\endgroup\ and\ \begingroup
  Wainer\endgroup(1983)Cicho\'n and Wainer]{cichon83}
Cicho\'n, E.A. and Wainer, S.S., 1983.
\newblock The slow-growing and the {G}rzegorczyk hierarchies.
\newblock \emph{Journal of Symbolic Logic}, 48\penalty0 (2):\penalty0 399--408.
\newblock \href{http://dx.doi.org/10.2307/2273557}
  {\nolinkurl{doi:10.2307/2273557}}.

\bibitem[\begingroup Cicho\'n\endgroup\ and\ \begingroup {Tahhan
  Bittar}\endgroup(1998)Cicho\'n and {Tahhan Bittar}]{cichon98}
Cicho\'n, E.A. and {Tahhan Bittar}, E., 1998.
\newblock Ordinal recursive bounds for {Higman}'s {T}heorem.
\newblock \emph{Theoretical Computer Science}, 201\penalty0 (1--2):\penalty0
  63--84.
\newblock \href{http://dx.doi.org/10.1016/S0304-3975(97)00009-1}
  {\nolinkurl{doi:10.1016/S0304-3975(97)00009-1}}.

\bibitem[\begingroup Cicho\'n\endgroup(2009)Cicho\'n]{lowhigman}
Cicho\'n, E.A., 2009.
\newblock Ordinal complexity measures.
\newblock {C}onference on Proofs and Computations in honour of Stan Wainer on
  the occasion of his 65th birthday.

\bibitem[\begingroup Clote\endgroup(1986)Clote]{clote}
Clote, P., 1986.
\newblock On the finite containment problem for {P}etri nets.
\newblock \emph{Theoretical Computer Science}, 43:\penalty0 99--105.
\newblock \href{http://dx.doi.org/10.1016/0304-3975(86)90169-6}
  {\nolinkurl{doi:10.1016/0304-3975(86)90169-6}}.

\bibitem[\begingroup de~Jongh\endgroup\ and\ \begingroup
  Parikh\endgroup(1977)de~Jongh and Parikh]{dejongh77}
de~Jongh, D.H.J. and Parikh, R., 1977.
\newblock Well-partial orderings and hierarchies.
\newblock \emph{Indagationes Mathematicae}, 39\penalty0 (3):\penalty0 195--207.
\newblock \href{http://dx.doi.org/10.1016/1385-7258(77)90067-1}
  {\nolinkurl{doi:10.1016/1385-7258(77)90067-1}}.

\bibitem[\begingroup Figueira\endgroup\ \textnormal{et~al.}(2010)Figueira,
  Hofman, and Lasota]{FHL10}
Figueira, D., Hofman, P., and Lasota, S., 2010.
\newblock Relating timed and register automata.
\newblock In Fr{\"o}schle, S. and Valencia, F., editors, \emph{EXPRESS
  2010}\natconfdetails[\emph{17th}]{\emph{{I}nternational {W}orkshop on
  {E}xpressiveness in {C}oncurrency}}, volume~41 of \emph{EPTCS}, pages 61--75.
\newblock \href{http://dx.doi.org/10.4204/EPTCS.41.5}
  {\nolinkurl{doi:10.4204/EPTCS.41.5}}.

\bibitem[\begingroup Figueira\endgroup\ \textnormal{et~al.}(2011)Figueira,
  Figueira, Schmitz, and Schnoebelen]{FFSS-arxiv10}
Figueira, D., Figueira, S., Schmitz, S., and Schnoebelen, {\relax Ph}., 2011.
\newblock {A}ckermannian and primitive-recursive bounds with {D}ickson's
  {L}emma.
\newblock In \emph{LICS 2011}\natconfdetails[\emph{26th}]{\emph{Annual IEEE
  Symposium on Logic in Computer Science}}, pages 269--278. IEEE.
\newblock \href{http://dx.doi.org/10.1109/LICS.2011.39}
  {\nolinkurl{doi:10.1109/LICS.2011.39}}.

\bibitem[\begingroup Finkel\endgroup\ and\ \begingroup
  Schnoebelen\endgroup(2001)Finkel and Schnoebelen]{FinSch-wsts}
Finkel, A. and Schnoebelen, {\relax Ph}., 2001.
\newblock Well-structured transition systems everywhere\string!
\newblock \emph{Theoretical Computer Science}, 256\penalty0 (1--2):\penalty0
  63--92.
\newblock \href{http://dx.doi.org/10.1016/S0304-3975(00)00102-X}
  {\nolinkurl{doi:10.1016/S0304-3975(00)00102-X}}.

\bibitem[\begingroup Henzinger\endgroup\ \textnormal{et~al.}(2005)Henzinger,
  Majumdar, and Raskin]{henzinger2005}
Henzinger, T.A., Majumdar, R., and Raskin, J.F., 2005.
\newblock A classification of symbolic transition systems.
\newblock \emph{ACM Transactions on Computational Logic}, 6\penalty0
  (1):\penalty0 1--32.
\newblock \href{http://dx.doi.org/10.1145/1042038.1042039}
  {\nolinkurl{doi:10.1145/1042038.1042039}}.

\bibitem[\begingroup Kruskal\endgroup(1972)Kruskal]{kruskal72}
Kruskal, J.B., 1972.
\newblock The theory of well-quasi-ordering\string: {A} frequently discovered
  concept.
\newblock \emph{Journal of Combinatorial Theory, Series A}, 13\penalty0
  (3):\penalty0 297--305.
\newblock \href{http://dx.doi.org/10.1016/0097-3165(72)90063-5}
  {\nolinkurl{doi:10.1016/0097-3165(72)90063-5}}.

\bibitem[\begingroup Lasota\endgroup\ and\ \begingroup
  Walukiewicz\endgroup(2008)Lasota and Walukiewicz]{ata}
Lasota, S. and Walukiewicz, I., 2008.
\newblock Alternating timed automata.
\newblock \emph{ACM Transactions on Computational Logic}, 9\penalty0
  (2):\penalty0 10.
\newblock \href{http://dx.doi.org/10.1145/1342991.1342994}
  {\nolinkurl{doi:10.1145/1342991.1342994}}.

\bibitem[\begingroup L\"ob\endgroup\ and\ \begingroup
  Wainer\endgroup(1970)L\"ob and Wainer]{fast}
L\"ob, M. and Wainer, S., 1970.
\newblock Hierarchies of number theoretic functions, {I}.
\newblock \emph{Archiv f{\"u}r Mathematische Logik und Grundlagenforschung},
  13:\penalty0 39--51.
\newblock \href{http://dx.doi.org/10.1007/BF01967649}
  {\nolinkurl{doi:10.1007/BF01967649}}.

\bibitem[\begingroup McAloon\endgroup(1984)McAloon]{mcaloon}
McAloon, K., 1984.
\newblock Petri nets and large finite sets.
\newblock \emph{Theoretical Computer Science}, 32\penalty0 (1--2):\penalty0
  173--183.
\newblock \href{http://dx.doi.org/10.1016/0304-3975(84)90029-X}
  {\nolinkurl{doi:10.1016/0304-3975(84)90029-X}}.

\bibitem[\begingroup Moser\endgroup\ and\ \begingroup
  Weiermann\endgroup(2003)Moser and Weiermann]{moser03}
Moser, G. and Weiermann, A., 2003.
\newblock Relating derivation lengths with the slow-growing hierarchy directly.
\newblock In Nieuwenhuis, R., editor, \emph{RTA
  2003}\natconfdetails[\emph{14th}]{\emph{International Conference on Rewriting
  Techniques and Applications}}, volume 2706 of \emph{Lecture Notes in Computer
  Science}, pages 296--310. Springer.
\newblock \href{http://dx.doi.org/10.1007/3-540-44881-0_21}
  {\nolinkurl{doi:10.1007/3-540-44881-0_21}}.

\bibitem[\begingroup Ouaknine\endgroup\ and\ \begingroup
  Worrell\endgroup(2007)Ouaknine and Worrell]{mtl}
Ouaknine, J.O. and Worrell, J.B., 2007.
\newblock On the decidability and complexity of {M}etric {T}emporal {L}ogic
  over finite words.
\newblock \emph{Logical Methods in Computer Science}, 3\penalty0 (1):\penalty0
  8.
\newblock \href{http://dx.doi.org/10.2168/LMCS-3(1:8)2007}
  {\nolinkurl{doi:10.2168/LMCS-3(1:8)2007}}.

\bibitem[\begingroup Schnoebelen\endgroup(2010)Schnoebelen]{phs-rp2010}
Schnoebelen, {\relax Ph}., 2010.
\newblock Lossy counter machines decidability cheat sheet.
\newblock In Ku\v{c}era, A. and Potapov, I., editors, \emph{RP
  2010}\natconfdetails[\emph{4th}]{\emph{Workshop on Reachability Problems}},
  volume 6227 of \emph{Lecture Notes in Computer Science}, pages 51--75.
  Springer.
\newblock \href{http://dx.doi.org/10.1007/978-3-642-15349-5_4}
  {\nolinkurl{doi:10.1007/978-3-642-15349-5_4}}.

\bibitem[\begingroup Touzet\endgroup(1997)Touzet]{THESE-Touzet97}
Touzet, H., 1997.
\newblock \emph{Propri\'et\'es combinatoires pour la terminaison de syst\`emes
  des r\'e\'ecriture}.
\newblock Th{\`e}se de doctorat, Universit\'e de Nancy 1, France.

\bibitem[\begingroup Touzet\endgroup(2002)Touzet]{touzet2002}
Touzet, H., 2002.
\newblock A characterisation of multiply recursive functions with {Higman}'s
  {L}emma.
\newblock \emph{Information and Computation}, 178\penalty0 (2):\penalty0
  534--544.
\newblock \href{http://dx.doi.org/10.1006/inco.2002.3160}
  {\nolinkurl{doi:10.1006/inco.2002.3160}}.

\bibitem[\begingroup Weiermann\endgroup(1994)Weiermann]{weiermann94}
Weiermann, A., 1994.
\newblock Complexity bounds for some finite forms of {K}ruskal's {T}heorem.
\newblock \emph{Journal of Symbolic Computation}, 18\penalty0 (5):\penalty0
  463--488.
\newblock \href{http://dx.doi.org/10.1006/jsco.1994.1059}
  {\nolinkurl{doi:10.1006/jsco.1994.1059}}.

\end{thebibliography}


\begin{thebibliography}{2}
\providecommand{\natexlab}[1]{#1}
\providecommand{\href}[2]{#2}
\providecommand{\nolinkurl}[1]{#1}
\providecommand{\natconfdetails}[2][]{}
\def\UrlBreaks{\do\@\do\\\do\/\do\!\do\|\do\;\do\>\do\]%
 \do\,\do\?\do\'\do+\do\=\do\#}
\providecommand{\urlstyle}[1]{}
\urlstyle{same}

\bibitem[\begingroup Hasegawa\endgroup(1994)Hasegawa]{hasegawa94}
Hasegawa, R., 1994.
\newblock Well-ordering of algebras and {Kruskal}'s theorem.
\newblock In Jones, N.D., Hagiya, M., and Sato, M., editors, \emph{Logic,
  Language and Computation, Festschrift in Honor of Satoru Takasu}, volume 792
  of \emph{Lecture Notes in Computer Science}, pages 133--172. Springer.
\newblock \href{http://dx.doi.org/10.1007/BFb0032399}
  {\nolinkurl{doi:10.1007/BFb0032399}}.

\bibitem[\begingroup Simpson\endgroup(1988)Simpson]{simpson88}
Simpson, S.G., 1988.
\newblock Ordinal numbers and the {H}ilbert {B}asis {T}heorem.
\newblock \emph{Journal of Symbolic Logic}, 53\penalty0 (3):\penalty0 961--974.
\newblock \href{http://dx.doi.org/10.2307/2274585}
  {\nolinkurl{doi:10.2307/2274585}}.

\end{thebibliography}
